\documentclass[nocopyrightspace,natbib]{sigplanconf}

\usepackage{amsthm,amssymb,amsmath,latexsym}

\usepackage{upgreek}

\usepackage[all]{xy}

\usepackage{stmaryrd}

\usepackage{code}

\usepackage[
  pdftitle={Push-down control-flow analysis of higher-order programs},
  pdfauthor={Christopher Earl, Matthew Might, and David Van Horn},
  pdfsubject={Computer Science},
  pdfkeywords={pushdown, cfa, dyck, lambda, calculus, analysis,
    abstract interpretation, polyvariant, functional, anf, pda, cesk},
  colorlinks=true,
  linkcolor=black,
  anchorcolor=black,
  citecolor=black,
  breaklinks=true
]
{hyperref}

\newcommand{\defterm}[1]{\textbf{#1}}

\newcommand{\ie}{\emph{i.e.}}
\newcommand{\eg}{\emph{e.g.}}

\newcommand{\syn}[1]{\mathsf{#1}}
\newcommand{\var}[1]{\mathit{#1}}
\newcommand{\s}[1]{\mathit{#1}}

\newcommand{\parto}{\rightharpoonup}
\newcommand{\dom}{\var{dom}}

\newcommand{\set}[1]{\left\{#1\right\}}
\newcommand{\setbuild}[2]{\left\{ #1 : #2\right\}}

\newcommand{\Pow}[1]{{\mathcal{P}\left(#1\right)}}
\newcommand{\PowSm}[1]{{\mathcal{P}(#1)}}

\newcommand{\union}{\cup}

\newcommand{\Union}{\bigcup}

\newcommand{\vect}[1]{\langle #1\rangle}
\newcommand{\vecp}[1]{\vec{#1}\;'}

\newcommand{\To}{\mathrel{\Rightarrow}}

\newcommand{\wt}{\sqsubseteq}

\newcommand{\join}{\sqcup}

\newcommand{\bigjoin}{\bigsqcup}

\DeclareMathOperator{\lfp}{lfp}

\newcommand{\transition}{\delta}

\newcommand{\QStates}{Q}
\newcommand{\FStates}{F}

\newcommand{\StackAlpha}{\Gamma}

\newcommand{\stackchar}{\gamma}

\newcommand{\sembr}[1]{\ensuremath{[\![{#1}]\!]}}

\newcommand{\opor}{\mathrel{|}}

\newcommand{\Alphabet}{A}

\newcommand{\produces}{\mathrel{::=}}

\newcommand{\Lang}{\mathcal{L}} % An interpretation function
\newcommand{\vv}{v}

\newcommand{\lam}{\ensuremath{\var{lam}}}

\newcommand{\lamterm}{$\lambda$-term}
\newcommand{\lc}{$\lambda$-calculus}

\newcommand{\call}{\ensuremath{\var{call}}}

\newcommand{\lt}[2]{\lambda #1.#2}

\newcommand{\ttlp}{\mbox{\tt (}}
\newcommand{\ttrp}{\mbox{\tt )}}

\newcommand{\appform}[2]{\ttlp #1\; #2\ttrp}
\newcommand{\lamform}[2]{\ttlp \uplambda\;\ttlp#1\ttrp\;#2\ttrp}

\newcommand{\letiform}[3]{\ttlp {\tt let}\; \ttlp\ttlp#1\; #2\ttrp\ttrp\; #3\ttrp}

\newcommand{\fexpr}{f}
\newcommand{\expr}{e}
\newcommand{\aexpr}{\mbox{\sl {\ae}}}

\newcommand{\Eval}{{\mathcal{E}}}
\newcommand{\ArgEval}{{\mathcal{A}}}

\newcommand{\Inject}{{\mathcal{I}}}

\newcommand{\QState}{Q}
\newcommand{\qstate}{q}

\newcommand{\tf}{f}

\newcommand{\store}{\sigma}

\newcommand{\env}{\rho}

\newcommand{\clo}{\var{clo}}
\newcommand{\cont}{\kappa}

\newcommand{\alloc}{\mathit{alloc}}

\newcommand{\addr}{a}

\newcommand{\aTo}{\leadsto}

\newcommand{\aInject}{{\hat{\mathcal{I}}}}

\newcommand{\sa}[1]{\widehat{\mathit{#1}}}

\newcommand{\aEval}{{\hat{\mathcal{E}}}}
\newcommand{\aArgEval}{{\hat{\mathcal{A}}}}

\newcommand{\atf}{{\hat{f}}}

\newcommand{\astore}{{\hat{\sigma}}}
\newcommand{\aenv}{{\hat{\rho}}}

\newcommand{\aclo}{{\widehat{\var{clo}}}}

\newcommand{\acont}{{\hat{\kappa}}}

\newcommand{\aaddr}{{\hat{\addr}}}
\newcommand{\aalloc}{{\widehat{alloc}}}

\newcommand{\absmap}{\alpha}
\newcommand{\abs}[1]{|#1|}

\newcommand{\ControlStates}{Q}
\newcommand{\transfunction}{\delta}

\newcommand{\conf}{c}
\newcommand{\aconf}{{\hat c}}

\newcommand{\phrame}{\phi}
\newcommand{\aphrame}{\hat{\phi}}

\newcommand{\stackact}{g}

\DeclareMathOperator*{\PDTrans}{\longmapsto}

\newcommand{\fDSG}{\mathcal{DSG}}

\newcommand{\afPDA}{\widehat{\mathcal{PDA}}}

\newcommand{\afRPDS}{\widehat{\mathcal{RPDS}}}

\renewcommand{\Alphabet}{\Sigma}

\DeclareMathOperator*{\pdedge}{\rightarrowtail}

\newcommand{\triedge}[3]{#1 \mathrel{\pdedge^{#2}} #3}
\newcommand{\biedge}[2]{#1 \mathrel{\rightarrowtail} #2}

\newcommand{\DSStates}{S}
\newcommand{\DSEdges}{E}
\newcommand{\DSFrames}{\StackAlpha}
\newcommand{\dsframe}{\stackchar}
\newcommand{\dsstate}{s}

\newcommand{\mkDSG}{\mathcal{F}}

\newcommand{\ecg}{$\epsilon$-closure graph}

\DeclareMathOperator*{\RPDTrans}{{\longmapsto\!\!\!\!\!\!\!\longrightarrow}}

\newcommand{\fsprout}{\mathit{sprout}}
\newcommand{\faddpush}{\mathit{addPush}}
\newcommand{\faddpop}{\mathit{addPop}}
\newcommand{\faddempty}{\mathit{addEmpty}}

\newcommand{\eancestor}[1]{\overleftarrow{G}_\epsilon[#1]}
\newcommand{\edescendent}[1]{\overrightarrow{G}_\epsilon[#1]}

\newcommand{\system}{C}
\newcommand{\asystem}{{\hat{C}}}

\newcommand{\apconf}{{\hat{\pi}}}
\newcommand{\apstate}{{\hat{\psi}}}

\DeclareMathOperator*{\afTo}{\rightarrowtriangle}

\DeclareMathOperator*{\apTo}{\rightharpoondown}

\newcommand{\pdcfato}[4]{#1 \mathrel{\apTo^{#2}_{#3}} #4}

\newcommand{\mylongtitle}{Pushdown Control-Flow Analysis of Higher-Order Programs}

\newcommand{\myshorttitle}{PDCFA}
\newcommand{\mytitle}{\mylongtitle}

\newcommand{\mytitlebanner}{In submission}
\newcommand{\footertitle}{\myshorttitle}

\newtheorem{theorem}{Theorem}[section]
\newtheorem{lemma}{Lemma}[section]

\begin{document}

\conferenceinfo{WXYZ '05}{date, City.} 
\copyrightyear{2005} 
\copyrightdata{[to be supplied]} 

\titlebanner{\mytitlebanner}        % These are ignored unless
\preprintfooter{\footertitle}       % 'preprint' option specified.

\title{\mytitle}
%\subtitle{\mysubtitle}

\authorinfo{Christopher Earl \and Matthew Might}
           {University of Utah}
           {\{cwearl,might\}@cs.utah.edu}

\authorinfo{David Van Horn\thanks{Supported by the National Science
    Foundation under grant 0937060 to the Computing Research
    Association for the CIFellow Project.}}
           {Northeastern University}
           {dvanhorn@ccs.neu.edu}

\maketitle

\begin{abstract}
Context-free approaches to static analysis gain precision over
classical approaches by perfectly matching returns to call sites---a
property that eliminates spurious interprocedural paths.  Vardoulakis
and Shivers's recent formulation of CFA2 showed that it is possible
(if expensive) to apply context-free methods to higher-order languages
and gain the same boost in precision achieved over first-order
programs.

To this young body of work on context-free analysis of higher-order
programs, we contribute a 
pushdown control-flow analysis framework, which we derive as an abstract
interpretation of a CESK machine with an unbounded stack.
One instantiation of this framework marks the first polyvariant
pushdown analysis of higher-order programs; another marks the first
polynomial-time analysis.
In the end, we arrive at a framework for control-flow analysis that can
efficiently compute pushdown generalizations of classical control-flow
analyses.
\end{abstract}

% \category{CR-number}{subcategory}{third-level}

% \terms
% term1, term2

% \keywords
% keyword1, keyword2

%%%%%%%%%%%%%%%%%%%%%%%%%%%%%%%%%%%%%%%%%%%%%
%% \input{content}
\section{Introduction}

Static analysis is bound by theory to accept diminished precision as
the price of decidability.
The easiest way to guarantee decidability for a static analysis is to
restrict it to a finite state-space.
Not surprisingly, finite state-spaces have become a crutch.

Whenever an abstraction maps an infinite (concrete) state-space down
to the finite state-space of a static analysis, the
pigeon-hole principle forces merging.
Distinct execution paths and values  can and do 
become indistinguishable under abstraction,
\eg{}, 3 and 4 both abstract to the same value: $\mathbf{positive}$.

Our message is that finite abstraction goes too far: we can abstract
into an infinite state-space to improve precision, yet remain decidable.
% ,
% and retain much of the classical control-flow analysis framework.
% %
% (By retaining parts of the classical framework, we can apply many of the
% techniques and optimizations designed for the classical framework to our
% new analysis.)
%
%\emph{
Specifically, we can abstract the concrete semantics of a
  higher-order language into a pushdown automaton (PDA).  As an
  infinite-state system, a PDA-based abstraction preserves more information than a classical finite-state analysis.  
Yet, being
  less powerful than a Turing machine, properties important for
  computing control-flow analysis (e.g. emptiness, intersection with
  regular languages, reachability) remain decidable. 
%
 % Our approach
 %  to this abstraction is able to make use of many optimizations designed for
 %  classical control-flow analysis.%}

\subsection{The problem with merging}
A short example provides a sense of how the inevitable merging that
occurs under a finite abstraction harms precision.
Shivers's 0CFA~\cite{mattmight:Shivers:1991:CFA} produces
spurious data-flows and return-flows in the following example:
\begin{code}
 (let* ((id (lambda (x) x)) 
        (a  (id 3))
        (b  (id 4)))  
   a)\end{code}
 0CFA says that the flow set for the variable {\tt a} contains both 3 and
 4.
In fact, so does the flow set for the variable {\tt b}.
For return-flow,\footnote{``Return-flow'' analysis asks to which
  call sites a given return point may return. In the presence of tail
  calls, which break the balance between calls and returns,
  return-flow analysis differs from control-flow analysis.}  0CFA says
that the invocation of {\tt (id 4)} may return to the invocation of {\tt
  (id 3)} or {{\tt (id 4)}} and vice versa; that is, according to
Shivers's 0CFA, this program contains a loop.

To combat merging, control-flow analyses have focused on increasing
context-sensitivity~\cite{mattmight:Shivers:1991:CFA}.
Context-sensitivity tries to qualify any answer that comes back from a
CFA with a context in which that answer is valid.
That is, instead of answering
``$\lambda_{42}$ may flow to variable $\vv_{13}$,''
a context-sensitive analysis might answer
``$\lambda_{42}$ may flow to variable $\vv_{13}$
  \emph{when
  bound after calling $f$}.''
While context-sensitivity recovers some lost precision, it is no
silver bullet.
A finite-state analysis stores only a finite amount of program context
to discriminate data- and control-flows during analysis.
Yet, the pigeon-hole principle remains merciless: as long as the
state-space is finite, merging is inevitable for some programs.

Of all the forms of merging, the most pernicious is the merging of
return-flow information.
% \footnote{The return-flow question is, ``To
%   which call site might a return expression return?''}
%
As the example shows, a finite-state control-flow analysis will lose
track of where return-points return once the maximum bound on context
is exceeded.
\emph{Even in programs with no higher-order functions}, return-flow merging
will still happen during control-flow analysis.

% In a finite-state setting, return-flow merging (or any kind of merging
% for that matter) is ultimately unavoidable.
% %
% Fortunately, there is a class of infinite state systems for which many
% important properties (such as reachability) remain decidable:
% pushdown automata (PDA).
% %
% Pushdown automata bear infinite state-spaces because they are a
% finite-state automaton adjoined with an unbounded stack.

% We aim to skirt along the boundary of decidability in static analysis
% by abstracting into a pushdown automaton instead of a finite-state
% machine.
% %
% In particular, we plan to model an abstraction of the program stack as
% the PDA's stack.
% %
% Because the program stack determines return-flow information, using
% the PDA's stack to model the program stack produces a higher-order
% control-flow analysis with \emph{effectively infinite return-flow
%   context-sensitivity}.

\subsection{A first shot: CFA2}
Vardoulakis and Shivers's recent work on
CFA2~\cite{dvanhorn:DBLP:conf/esop/VardoulakisS10} constitutes an opening
salvo on ending the return-flow problem for the static analysis of
higher-order programs.
CFA2 employs an implicit pushdown system that models the stack of a
program.
CFA2 solves the return-flow problem for higher-order programs, but it
has drawbacks:
\begin{enumerate}

\item CFA2 allows only monovariant precision.

\item CFA2 has exponential complexity in the size of the program.

\item CFA2 is restricted to continuation-passing style.

\end{enumerate}

Our solution overcomes all three drawbacks: it allows polyvariant
precision, we can widen it to $O(n^6)$-time complexity in the
monovariant case and we can operate on direct-style programs.

\subsection{Our solution: Abstraction to pushdown systems}
To prevent return-flow merging during higher-order control-flow
analysis, we abstract into an explicit pushdown system
instead of a finite-state machine.
The program stack, which determines return-flow, will remain unbounded
and become the pushdown stack.
As a result, return-flow information will never be merged: in the
abstract semantics, a function returns only whence it was called.

% Our development will begin with a CESK
% machine~\cite{mattmight:Felleisen} for A-Normal Form
% \lc~\cite{mattmight:Flanagan}.
% %
% We abstract this machine, but leave its stack component unbounded.
% %
% The remainder of our work consists of finding a way to efficiently
% compute state-reachability for this abstracted CESK machine.
% %
% So, we develop a series of state-reachability algorithms; the theme to
% this series is that with each new algorithm, time-complexity falls
% while implementation complexity rises.
% %
% By the end of this series, we will have a polynomial-time pushdown
% control-flow analysis for higher-order programs.

\subsection{Overview}

This paper is organized as follows: first, we define a variant of the
CESK machine~\cite{mattmight:Felleisen:1987:CESK} for the A-Normal
Form \lc{}~\cite{mattmight:Flanagan:1993:ANF}.
In performing analysis, we wish to soundly approximate intensional
properties of this machine when it evaluates a given program.
To do so, we construct an abstract interpretation of the machine.
The abstracted CESK machine operates much like its concrete
counterpart and soundly approximates its behavior, but crucially, many
properties of the concrete machine that are undecidable become
decidable when considered against the abstracted machine (e.g.~``is a
given machine configuration reachable?'' becomes a decidable property).
 
The abstract counterpart to the CESK machine is constructed by
bounding the store component of the machine to some finite size.
However, the stack component (represented as a continuation) is left
unabstracted.  
(This is in contrast to finite-state abstractions that store-allocate
continuations~\cite{dvanhorn:VanHorn2010Abstracting}.)
Unlike most higher-order abstract interpreters, the unbounded stack
implies this machine has a potentially infinite set of reachable
machine configurations, and therefore enumerating them is not a
feasible approach to performing analysis.

Instead, we demonstrate how properties can be decided by transforming
the abstracted CESK machine into an equivalent pushdown automaton.
We then reduce higher-order control-flow analysis to deciding the
non-emptiness of a language derived from the PDA.
(This language happens to be the intersection of a regular language
and the context-free language described by the PDA.)
This approach---though concise, precise and decidable---is formidably
expensive, with complexity doubly exponential in the size of the
program.

We simplify the algorithm to merely exponential in the size of the
input program by reducing the control-flow problem to pushdown
reachability~\cite{mattmight:Bouajjani:1997:PDA-Reachability}.
Unfortunately, the abstracted CESK machine has an exponential number
of control states with respect to the size of the program.
Thus, pushdown reachability for higher-order programs appears to be
inherently exponential.

% Existing techniques for pushdown reachability are polynomial-time, but
% they are \emph{polynomial-time in the number of control states in the
%   PDA}, and unfortunately, the abstracted CESK machine has an
% exponential number of control states with respect to the size of the
% program.
% %
% Thus, existing techniques for pushdown reachability produce an
% algorithm that is exponential (even in the best case) in the size of
% the input program.

Noting that most control states in the abstracted CESK machine are
actually unreachable, we present a fixed-point algorithm for deciding
pushdown reachability that is polynomial-time in the number of
\emph{reachable} control states.
Since the pushdown systems produced by abstracted CESK machines are
sparse, such algorithms, though exponential in the worst case, are
reasonable options.
Yet, we can do better.

Next, we add an \ecg{} (a graph encoding no-stack-change reachability)
and a work-list to the fixed-point algorithm.
Together, these lower the cost of finding
the reachable states of a pushdown system from
$O(\abs{\StackAlpha}^4m^5)$ to $O(\abs{\StackAlpha}^2m^4)$, where
$\StackAlpha$ is the stack alphabet and $m$ is the number of reachable
control states.

To drop the complexity of our analysis to polynomial-time in the size
of the input program, we must resort to both widening and
monovariance.
Widening with a single-threaded store and using a monovariant
allocation strategy yields a pushdown control-flow analysis with
polynomial-time complexity, at $O(n^6)$, where $n$ is the size of the
program.

Finally, we briefly highlight applications of pushdown control-flow
analyses that are outside the reach of classical ones, discuss related
work, and conclude.
% %
% We point to dependence analysis and escape analysis, since both depend
% on the structure of the program stack.

% \subsection{Contributions}

% We contribute 
% \begin{enumerate}
% \item the first \emph{pushdown} higher-order control-flow analysis with 
% an explicitly-defined pushdown system;

% \item the first context-free control-flow analysis for direct-style;

% \item the first polyvariant, context-free control-flow analysis;

% \item the first polynomial-time, context-free control-flow analysis;

% \item the notion of Dyck state graphs; and

% \item two algorithms for finding the reachable control states of a
%   pushdown system whose complexities are a function of \emph{reachable}
%     control states rather than \emph{all} control states.
% \end{enumerate}
% \paragraph{Note} Terms like ``context-free'' and ``reachability'' have been used in
% conjunction with terms like ``0CFA'' before, but please note that we
% are not using CFL-reachability to compute a classical control-flow
% analysis like 0CFA or $k$-CFA; the analyses we are computing are
% fundamentally richer. (See Section~\ref{sec:related} for discussion.)

\section{Pushdown preliminaries}

In this work, we make use of both pushdown systems and pushdown
automata.
(A pushdown automaton is a specific kind of pushdown system.)
There are many (equivalent) definitions of these machines in the
literature, so we adapt our own definitions from \cite{mattmight:Sipser:2005:Theory}.
Even those familiar with pushdown theory may want to skim this
section to pick up our notation.

\subsection{Syntactic sugar}

When a triple $(x,\ell,x')$ is an edge in a labeled graph, a little
syntactic sugar aids presentation:
\begin{equation*}
  x \pdedge^\ell x'  \equiv
  (x,\ell,x')
  \text.
\end{equation*}
Similarly, when a pair $(x,x')$ is a graph edge:
\begin{equation*}
  \biedge{x}{x'} \equiv (x,x')
  \text.
\end{equation*}
We use both
string and vector notation for sequences:
\begin{equation*}
  a_1 a_2 \ldots a_n \equiv \vect{a_1,a_2,\ldots,a_n}
  \equiv
  \vec{a} 
  \text.
\end{equation*}

\subsection{Stack actions, stack change and stack manipulation}

Stacks are sequences over a stack alphabet $\StackAlpha$.
Pushdown systems do much stack manipulation, so to represent this
more concisely, we turn stack alphabets into ``stack-action'' sets;
each character represents a change to the stack: push, pop or no
change.

For each character $\stackchar$ in a stack alphabet $\StackAlpha$, the
\defterm{stack-action} set $\StackAlpha_\pm$ contains a push character
$\stackchar_{+}$ and a pop character $\stackchar_{-}$; it also
contains a no-stack-change indicator, $\epsilon$:
\begin{align*}
  \stackact \in \StackAlpha_\pm &\produces \epsilon && \text{[stack unchanged]} 
  \\
  &\;\;\opor\;\; \stackchar_{+}  \;\;\;\text{ for each } \stackchar \in \StackAlpha && \text{[pushed $\stackchar$]}
  \\
  &\;\;\opor\;\; \stackchar_{-}  \;\;\;\text{ for each } \stackchar \in \StackAlpha && \text{[popped $\stackchar$]}
  \text.
\end{align*}
In this paper, the symbol $\stackact$ represents some stack action.

% Given a string of stack actions, we can compact it into a minimal
% string describing net stack change.
% %
% We do so through the operator $\fnet{\cdot} : \StackAlpha_\pm^* \to
% \StackAlpha_\pm^*$, which cancels out opposing adjacent push-pop stack
% actions:
% \begin{align*}
%   \fnet{\vec{\stackact} \; \stackchar_+\stackchar_- \; \vecp{\stackact}} &= 
%   \fnet{\vec{\stackact} \; \vecp{\stackact}} 
%   &
%   \fnet{\vec{\stackact} \; \epsilon \; \vecp{\stackact}} &= 
%   \fnet{\vec{\stackact} \; \vecp{\stackact}} 
%   \text,
% \end{align*}
% so that:
% \begin{align*}
%   \fnet{\vec{\stackact}} &= \vec{\stackact}\text,
% \end{align*}
% if there are no cancellations to be made in the string $\vec{\stackchar}$.

% We can convert a net string back into a stack by stripping off the
% push symbols with the stackify operator, $\fstackify{\cdot} :
% \StackAlpha^{*}_\pm \parto \StackAlpha^*$:
% \begin{align*}
%   \fstackify{\stackchar_+ \stackchar_+' \ldots \stackchar_+^{(n)}} =
%   \vect{\stackchar^{(n)}, \ldots, \stackchar', \stackchar}
%  \text.
% \end{align*}
% Clearly, the stackify operator is defined only for strings containing
% all push actions. 

% The function $\fTop : \StackAlpha^* \parto \StackAlpha$ returns the
% top of a stack (if any):
% \begin{equation*}
%   \fTop\vect{\stackchar_1,\ldots,\stackchar_n} = \stackchar_1
%   \text.
% \end{equation*}

% (Might's 
% work
% on $\Delta$CFA~\cite{mattmight:Might:2006:DeltaCFA,mattmight:Might:2007:DeltaCFA,mattmight:Might:2007:Dissertation},
% studies the theory of these ``frame strings'' in detail.)

\subsection{Pushdown systems}
A \defterm{pushdown system} is a triple
$M = (\ControlStates,\StackAlpha,\transfunction)$ where:
\begin{enumerate}

\item $\ControlStates$ is a finite set of control states;

\item $\StackAlpha$ is a stack alphabet; and

\item $\transfunction \subseteq
  \ControlStates \times \StackAlpha_\pm \times \ControlStates$ is a transition relation.
\end{enumerate}
We use $\mathbb{PDS}$ to denote the class of all pushdown systems.

Unlike the more widely known pushdown automaton, a pushdown system
\emph{does not recognize a language}.
\\

\noindent
For the following definitions, let $M = (\ControlStates,\StackAlpha,\transfunction)$.
\begin{itemize}

\item The \defterm{configurations} of this
  machine---$\s{Configs}(M)$---are pairs over control states and
  stacks:
  \begin{equation*}
    \s{Configs}(M) = \ControlStates \times \StackAlpha^*
    \!\!\!
    \text.
  \end{equation*}

\item The labeled \defterm{transition relation} $(\PDTrans_{M}) \subseteq \s{Configs}(M) \times \StackAlpha_\pm \times 
  \s{Configs}(M)$ determines whether one configuration may transition to another while performing the given stack action:
  \begin{align*}
    (\qstate, \vec{\stackchar}) 
    \mathrel{\PDTrans_M^\epsilon}
    (\qstate',\vec{\stackchar}) 
    & \text{ iff }
    %(\qstate,\epsilon,\qstate')
    \qstate \pdedge^\epsilon \qstate'
    \in \transfunction
    && \text{[no change]}
\\
   (\qstate, \stackchar : \vec{\stackchar}) 
   \mathrel{\PDTrans_M^{\stackchar_{-}}}
   (\qstate',\vec{\stackchar})
    & \text{ iff }
    %(\qstate,\stackchar_{-},\qstate') 
    \qstate \pdedge^{\stackchar_{-}} \qstate'
    \in \transfunction
    && \text{[pop]}
\\
    (\qstate, \vec{\stackchar}) 
    \mathrel{\PDTrans_{M}^{\stackchar_{+}}}
    (\qstate',\stackchar : \vec{\stackchar}) 
    & \text{ iff }
    %(\qstate,\stackchar_{+},\qstate')
    \qstate \pdedge^{\stackchar_{+}} \qstate'
    \in \transfunction
    && \text{[push]}
    \text.
  \end{align*}

\item If unlabelled, the transition relation $(\PDTrans)$ checks whether \emph{any} stack action can enable the transition:
 \begin{align*}
    \conf \mathrel{\PDTrans_{M}} \conf' \text{ iff }
    \conf \mathrel{\PDTrans_{M}^{\stackact}} \conf' \text{ for some stack action } \stackact 
    \text.
  \end{align*}

\item

For a string of stack actions $\stackact_1 \ldots
  \stackact_n$:
\begin{equation*}
  \conf_0 \mathrel{\PDTrans_M^{\stackact_1\ldots\stackact_n}} \conf_n
  \text{ iff }
  \conf_0
  \mathrel{\PDTrans_M^{\stackact_1}} 
  \conf_1
  \mathrel{\PDTrans_M^{\stackact_2}} 
  \cdots
  \mathrel{\PDTrans_M^{\stackact_{n-1}}} 
  \conf_{n-1} 
  \mathrel{\PDTrans_M^{\stackact_n}} 
  \conf_n\text,
\end{equation*}
for some configurations $\conf_0,\ldots,\conf_n$.

\item

For the transitive closure:
\begin{equation*}
  \conf \mathrel{\PDTrans_M^{*} \conf'
  \text{ iff }
  \conf \mathrel{\PDTrans_M^{\vec{\stackact}}} \conf'
  \text{ for some action string }
  \vec{\stackact}} 
  \text.
\end{equation*}

\end{itemize}

% A \defterm{legal configuration-path} through system $M$ is a sequence $\vec{\conf} =
% \vect{\conf_1,\ldots,\conf_n}$ where adjacent configurations are a
% legal transition:
% \begin{equation*}
%   \conf_i \mathrel{ \PDTrans_M} \conf_{i + 1}
%   \text.
% \end{equation*}
%

\paragraph{Note}
Some texts define the transition relation $\transfunction$ so that 
$\transfunction \subseteq \ControlStates \times \StackAlpha \times \ControlStates
\times \StackAlpha^*$.
In these texts, $(\qstate,\stackchar,\qstate',\vec{\stackchar}) \in
\transfunction$ means, ``if in control state $\qstate$ while the
character $\stackchar$ is on top, pop the stack, transition to
control state $\qstate'$ and push $\vec{\stackchar}$.''
Clearly, we can convert between these two representations by
introducing extra control states to our representation when it needs
to push multiple characters.

\subsection{Rooted pushdown systems}

A \defterm{rooted pushdown system} is a quadruple 
$(\ControlStates,\StackAlpha,\transfunction,\qstate_0)$ in which 
$(\ControlStates,\StackAlpha,\transfunction)$ is a pushdown system and
$\qstate_0 \in \ControlStates$ is an initial (root) state.
$\mathbb{RPDS}$  is the class of all rooted pushdown
systems.

For a rooted pushdown system $M =
(\ControlStates,\StackAlpha,\transfunction,\qstate_0)$, we define a
the \defterm{root-reachable transition relation}:
\begin{equation*}
  \conf \RPDTrans_M^{\stackact} \conf' \text{ iff } 
  (\qstate_0,\vect{})
  \mathrel{\PDTrans_M^*} 
  \conf 
  \text{ and }
  \conf 
  \mathrel{\PDTrans_M^\stackact}
  \conf'
  \text.
\end{equation*}
In other words, the root-reachable transition relation also makes
sure that the root control state can actually reach the transition.

We overload the root-reachable transition relation to operate on
control states as well:
\begin{equation*}
  \qstate 
   \mathrel{\RPDTrans_M^\stackact}
  \qstate' \text{ iff }
  (\qstate,\vec{\stackchar}) 
   \mathrel{\RPDTrans_M^\stackact}
  (\qstate',\vecp{\stackchar}) 
  \text{ for some stacks }
  \vec{\stackchar},
  \vecp{\stackchar}
  \text.
\end{equation*}
For both root-reachable relations, if we elide the stack-action label,
then, as in the un-rooted case, the transition holds if \emph{there
  exists} some stack action that enables the transition:
\begin{align*}
  \qstate 
   \mathrel{\RPDTrans_M}
  \qstate' \text{ iff }
  \qstate 
   \mathrel{\RPDTrans_M^{\stackact}}
  \qstate' 
  \text{ for some action } \stackact
  \text.
\end{align*}

\subsection{Pushdown automata}
A \defterm{pushdown automaton} is a generalization
of a rooted pushdown system, a 7-tuple
$(\QStates,\Alphabet,\StackAlpha,\transition,\qstate_0,\FStates,\vec{\stackchar})$ in which:
\begin{enumerate}
\item $\Alphabet$ is an input alphabet; 

\item $\transfunction \subseteq \ControlStates \times \StackAlpha_\pm
  \times (\Alphabet \union \set{\epsilon}) \times \ControlStates$ is a
  transition relation; 

\item $\FStates \subseteq \QStates$ is a set of accepting states; and

\item $\vec{\stackchar} \in \StackAlpha^*$ is the initial stack.

\end{enumerate}
We use $\mathbb{PDA}$ to denote the class of all pushdown automata.

Pushdown automata recognize languages over their input alphabet.
To do so, their transition relation may optionally consume an input
character upon transition.
Formally, a PDA $M = (\QStates,\Alphabet,\StackAlpha,\transition,\qstate_0,\FStates,\vec{\stackchar})$
recognizes the language $\Lang(M) \subseteq \Alphabet^*$:
\begin{align*}
\epsilon \in \Lang(M) &\text{ if }
\qstate_0 \in \FStates
% \\
% w \in \Lang(M) &\text{ if }
% \transition(\qstate_0,\stackchar_+,\epsilon,\qstate')
% \text{ and }
% w \in \Lang(\qstate',\Alphabet,\StackAlpha,\transition,\qstate, \FStates, \stackchar : \vec{\stackchar})
% \\
% w \in \Lang(M) &\text{ if }
% \transition(\qstate_0,\epsilon,\epsilon,\qstate')
% \text{ and }
% w \in \Lang(\qstate',\Alphabet,\StackAlpha,\transition,\qstate, \FStates, \vec{\stackchar})
% \\
% w \in \Lang(M) &\text{ if }
% \transition(\qstate_0,\stackchar_-,\epsilon,\qstate')
% \text{ and }
% w \in \Lang(\qstate',\Alphabet,\StackAlpha,\transition,\qstate, \FStates, \vec{\stackchar}')
% \\
% &\text{ where } \vec{\stackchar} = \vect{\stackchar_1,\stackchar_2,\ldots,\stackchar_n}
% \text{ and } \vec{\stackchar}' = \vect{\stackchar_2,\ldots,\stackchar_n}
%
\\
aw \in \Lang(M) &\text{ if }
\transition(\qstate_0,\stackchar_+,a,\qstate')
\text{ and }
w \in \Lang(\QState,\Alphabet,\StackAlpha,\transition,\qstate', \FStates, \stackchar : \vec{\stackchar})
\\
aw \in \Lang(M) &\text{ if }
\transition(\qstate_0,\epsilon,a,\qstate')
\text{ and }
w \in \Lang(\QState,\Alphabet,\StackAlpha,\transition,\qstate', \FStates, \vec{\stackchar})
\\
aw \in \Lang(M) &\text{ if }
\transition(\qstate_0,\stackchar_-,a,\qstate')
\text{ and }
w \in \Lang(\QState,\Alphabet,\StackAlpha,\transition,\qstate', \FStates, \vec{\stackchar}')
\\
&\text{ where } \vec{\stackchar} = \vect{\stackchar,\stackchar_2,\ldots,\stackchar_n}
\text{ and } \vec{\stackchar}' = \vect{\stackchar_2,\ldots,\stackchar_n}
\text,
\end{align*}
where $a$ is either the empty string $\epsilon$ or a single character.

%\subsection{Dyck languages}

\section{Setting: A-Normal Form 
  \texorpdfstring{$\boldsymbol\lambda$}{Lambda}-calculus}
\label{sec:anf}

Since our goal is to create pushdown control-flow analyses of
\emph{higher-order languages}, we choose to operate on the
\lc{}.
For simplicity of the concrete and abstract semantics, we choose to
analyze programs in A-Normal Form, however this is strictly a cosmetic
choice; all of our results can be replayed \emph{mutatis mutandis} in
a direct-style setting.
ANF enforces an order of evaluation and it requires
that all arguments to a function be atomic:
\begin{align*}
 \expr \in \syn{Exp} &\produces \letiform{\vv}{\call}{\expr} && \text{[non-tail call]}
 \\
 &\;\;\opor\;\; \call && \text{[tail call]}
 \\
 &\;\;\opor\;\; \aexpr && \text{[return]}
 \\
 \fexpr,\aexpr \in \syn{Atom} &\produces \vv \opor \lam && \text{[atomic expressions]}
 \\
 \lam \in \syn{Lam} &\produces \lamform{\vv}{\expr} && \text{[lambda terms]}
 \\
 \call \in \syn{Call} &\produces \appform{\fexpr}{\aexpr} && \text{[applications]}
 \\
 \vv \in \syn{Var} &\text{ is a set of identifiers} && \text{[variables]}
 \text.
\end{align*}

We use the CESK machine of \citet{mattmight:Felleisen:1987:CESK} to specify the semantics of
ANF.
We have chosen the CESK machine because it has an explicit stack, and
under abstraction, the stack component of our CESK machine will become
the stack component of a pushdown system.

First, we define a set of configurations ($\s{Conf}$) for this machine:
\begin{align*}
  \conf \in \s{Conf} &= \syn{Exp} \times \s{Env} \times \s{Store} \times \s{Kont} && \text{[configurations]}
  \\
  \env \in \s{Env} &= \syn{Var} \parto \s{Addr} && \text{[environments]}
  \\
  \store \in \s{Store} &= \s{Addr} \to \s{Clo} && \text{[stores]}
  % \\
  % \den \in \s{Den} &= \s{Clo} && \text{[denotable values]}
  \\
  \clo \in \s{Clo} &= \syn{Lam} \times \s{Env} && \text{[closures]}
  \\
  \cont \in \s{Kont} &= \s{Frame}^* && \text{[continuations]}
  \\
  \phrame \in \s{Frame} &= \syn{Var} \times \syn{Exp} \times \s{Env}  && \text{[stack frames]}
  \\
  \addr \in \s{Addr} &\text{ is an infinite set of addresses} && \text{[addresses]}
  \text.
\end{align*}
To define the semantics, we need five items:
\begin{enumerate}
\item $\Inject : \syn{Exp} \to \s{Conf}$ injects an expression into
  a configuration.

\item $\ArgEval : \syn{Atom} \times \s{Env} \times \s{Store} \parto
  \s{Clo}$ evaluates atomic expressions.

\item $\Eval : \syn{Exp} \to \Pow{\s{Conf}}$ computes the set of
  reachable machine configurations for a given program.

\item $(\To) \subseteq \s{Conf} \times \s{Conf}$ transitions between configurations.

\item $\alloc : \syn{Var} \times \s{Conf} \to \s{Addr}$ 
  chooses fresh store addresses for newly bound variables.

\end{enumerate}

\paragraph{Program injection}
The program injection function pairs an expression with an empty environment,
an empty store and an empty stack to create the initial configuration:
\begin{equation*}
  \conf_0 = \Inject(\expr) = (\expr, [], [], \vect{})
  \text.
\end{equation*}

\paragraph{Atomic expression evaluation}
The atomic expression evaluator, $\ArgEval : \syn{Atom} \times \s{Env}
\times \s{Store} \parto \s{Clo}$, returns the value of an atomic expression
in the context of an environment and a store:
\begin{align*}
  \ArgEval(\lam,\env,\store) &= (\lam,\env) && \text{[closure creation]}
  \\
  \ArgEval(\vv,\env,\store) &= \store(\env(\vv)) && \text{[variable look-up]}
  \text.
\end{align*}

\paragraph{Reachable configurations}
The evaluator $\Eval : \syn{Exp} \to \Pow{\s{Conf}}$ returns all 
configurations reachable from the initial configuration:
\begin{equation*}
  \Eval(\expr) = \setbuild{ \conf }{ \Inject(\expr) \To^* \conf } 
  \text.
\end{equation*}

\paragraph{Transition relation}
To define the transition $\conf \To \conf'$, we need three rules.
The first rule handles tail calls by evaluating the function into a
closure, evaluating the argument into a value and then moving to the
body of the \lamterm{} within the closure:
\begin{align*}
  \overbrace{(\sembr{\appform{\fexpr}{\aexpr}}, \env, \store, \cont)}^{\conf}
  &\To
  \overbrace{(\expr,\env'',\store',\cont)}^{\conf'}
  \text{, where }
  \\
  (\sembr{\lamform{\vv}{\expr}}, \env') &= \ArgEval(\fexpr,\env,\store)
  \\
  \addr &= \alloc(\vv,\conf)
  \\
  \env'' &= \env'[\vv \mapsto \addr]
  \\
  \store' &= \store[\addr \mapsto \ArgEval(\aexpr,\env,\store)]
  \text.
\end{align*}

\noindent
Non-tail call pushes a frame onto the stack and evaluates the call:
% \begin{align*}
%   \overbrace{(\sembr{\letiform{\vv}{\appform{\fexpr}{\aexpr}}{\expr}}, \env, \store, \cont)}^{\conf}
%   &\To
%   \overbrace{(\expr',\env'',\store',\phrame : \cont)}^{\conf'}
%   \text{, where }
%   \\
%   (\sembr{\lamform{\vv'}{\expr'}}, \env') &= \ArgEval(\fexpr,\env,\store)
%   \\
%   \addr &= \alloc(\vv,\state)
%   \\
%   \env'' &= \env'[\vv' \mapsto \addr]
%   \\
%   \store' &= \store[\addr \mapsto \ArgEval(\aexpr,\env,\store)]
%   \\
%   \phrame &= (\vv,\expr,\env)
%   \text.
% \end{align*}
\begin{align*}
  \overbrace{(\sembr{\letiform{\vv}{\call}{\expr}}, \env, \store, \cont)}^{\conf}
  &\To
  \overbrace{(\call,\env,\store, (\vv,\expr,\env) : \cont)}^{\conf'}
  \text.
\end{align*}

\noindent
Function return pops a stack frame:
\begin{align*}
  \overbrace{(\aexpr, \env, \store, (\vv,\expr,\env') : \cont)}^{\conf}
  &\To
  \overbrace{(\expr,\env'',\store', \cont)}^{\conf'}
  \text{, where }
  \\
  \addr &= \alloc(\vv,\conf)
  \\
  \env'' &= \env'[\vv \mapsto \addr]
  \\
  \store' &= \store[\addr \mapsto \ArgEval(\aexpr,\env,\store)]
  \text.  
\end{align*}

\paragraph{Allocation}
The address-allocation function is an opaque parameter in this
semantics.
We have done this so that the forthcoming abstract semantics may also
parameterize allocation, and in so doing provide a knob to tune the
polyvariance and context-sensitivity of the resulting analysis.
For the sake of defining the concrete semantics, letting addresses be
natural numbers suffices, and then the allocator can choose the lowest
unused address:
\begin{align*}
  \s{Addr} &= \mathbb{N}
  \\
  \alloc(v,(\expr,\env,\store,\cont)) &= 
  1 + \max(\dom(\store))
  \text.
\end{align*}

\section{An infinite-state abstract interpretation}
\label{sec:abstraction}

Our goal is to statically bound the higher-order control-flow of the
CESK machine of the previous section.
So, we are going to conduct an abstract interpretation.

Since we are concerned with return-flow precision, we are going to
abstract away less information than we normally would.
Specifically, we are going to construct an infinite-state abstract
interpretation of the CESK machine by leaving its stack unabstracted.
(With an infinite-state abstraction, the usual approach for computing
the static analysis---exploring the abstract configurations reachable
from some initial configuration---simply will not work.
Subsequent sections focus on finding an algorithm that can
compute a finite representation of the reachable abstract
configurations of the abstracted CESK machine.)

For the abstract interpretation of the CESK machine, we need an
abstract configuration-space (Figure~\ref{fig:abs-conf-space}).
To construct one, we force addresses to be a finite set, but
crucially, we leave the stack untouched.
When we compact the set of addresses into a finite set, the machine
may run out of addresses to allocate, and when it does, the
pigeon-hole principle will force multiple closures to reside at the
same address.
As a result, we have no choice but to force the range of the store to
become a power set in the abstract configuration-space.
To construct the abstract transition relation, we need five components
analogous to those from the concrete semantics.

\begin{figure}
%% \begin{small}
\begin{align*}
  \aconf \in \sa{Conf} &= \syn{Exp} \times \sa{Env} \times \sa{Store} \times \sa{Kont} && \text{[configurations]}
  \\
  \aenv \in \sa{Env} &= \syn{Var} \parto \sa{Addr} && \text{[environments]}
  \\
  \astore \in \sa{Store} &= \sa{Addr} \to \Pow{\sa{Clo}} && \text{[stores]}
  % \\
  % \aden \in \sa{Den} &= \Pow{\sa{Clo}} && \text{[denotable values]}
  \\
  \aclo \in \sa{Clo} &= \syn{Lam} \times \sa{Env} && \text{[closures]}
  \\
  \acont \in \sa{Kont} &= \sa{Frame}^* && \text{[continuations]}
  \\
  \aphrame \in \sa{Frame} &= \syn{Var} \times \syn{Exp} \times \sa{Env}  && \text{[stack frames]}
  \\
  \aaddr \in \sa{Addr} &\text{ is a \emph{finite} set of addresses} && \text{[addresses]}
  \text.
\end{align*}
\caption{The abstract configuration-space.}
\label{fig:abs-conf-space}
%% \end{small}
\end{figure}

\paragraph{Program injection}
The abstract injection function $\aInject : \syn{Exp} \to \sa{Conf}$
pairs an expression with an empty environment, an empty store and an
empty stack to create the initial abstract configuration:
\begin{equation*}
  \aconf_0 = \aInject(\expr) = (\expr, [], [], \vect{})
  \text.
\end{equation*}

\paragraph{Atomic expression evaluation}
The abstract atomic expression evaluator, $\aArgEval : \syn{Atom}
\times \sa{Env} \times \sa{Store} \to \PowSm{\sa{Clo}}$, returns the value of
an atomic expression in the context of an environment and a store;
note how it returns a set:
\begin{align*}
  \aArgEval(\lam,\aenv,\astore) &= \set{(\lam,\env)} && \text{[closure creation]}
  \\
  \aArgEval(\vv,\aenv,\astore) &= \astore(\aenv(\vv)) && \text{[variable look-up]}
  \text.
\end{align*}

\paragraph{Reachable configurations}
The abstract program evaluator $\aEval : \syn{Exp} \to
\PowSm{\sa{Conf}}$ returns all of the configurations reachable from
the initial configuration:
\begin{equation*}
  \aEval(\expr) = \setbuild{ \aconf }{ \aInject(\expr) \aTo^* \aconf } 
  \text.
\end{equation*}
Because there are an infinite number of abstract configurations, a
na\"ive implementation of this function may not terminate.
In Sections~\ref{sec:pda} through \ref{sec:ecg-worklist}, we show that
there is a way to compute a finite representation of this set.

\paragraph{Transition relation}
The abstract transition relation $(\aTo) \subseteq \sa{Conf} \times
\sa{Conf}$ has three rules, one of which has become nondeterministic.
A tail call may fork because there could be multiple abstract closures
that it is invoking:
\begin{align*}
  \overbrace{(\sembr{\appform{\fexpr}{\aexpr}}, \aenv, \astore, \acont)}^{\aconf}
  &\aTo
  \overbrace{(\expr,\aenv'',\astore',\acont)}^{\aconf'}
  \text{, where }
  \\
  (\sembr{\lamform{\vv}{\expr}}, \aenv') &\in \aArgEval(\fexpr,\aenv,\astore)
  \\
  \aaddr &= \aalloc(\vv,\aconf)
  \\
  \aenv'' &= \aenv'[\vv \mapsto \aaddr]
  \\
  \astore' &= \astore \join [\aaddr \mapsto \aArgEval(\aexpr,\aenv,\astore)]
  \text.
\end{align*}
We define all of the partial orders shortly, but for stores:
\begin{equation*}
  (\astore \join \astore')(\aaddr) = \astore(\aaddr) \union \astore'(\aaddr)
  \text.
\end{equation*}

\noindent
A non-tail call pushes a frame onto the stack and evaluates the call:
\begin{align*}
  \overbrace{(\sembr{\letiform{\vv}{\call}{\expr}}, \aenv, \astore, \acont)}^{\aconf}
  &\aTo
  \overbrace{(\call,\aenv,\astore, (\vv,\expr,\aenv) : \acont)}^{\aconf'} 
  \text.
\end{align*}

\noindent
A function return pops a stack frame:
\begin{align*}
  \overbrace{(\aexpr, \aenv, \astore, (\vv,\expr,\aenv') : \acont)}^{\aconf}
  &\aTo
  \overbrace{(\expr,\aenv'',\astore', \acont)}^{\aconf'}
  \text{, where }
  \\
  \aaddr &= \aalloc(\vv,\aconf)
  \\
  \aenv'' &= \aenv'[\vv \mapsto \aaddr]
  \\
  \astore' &= \astore \join [\aaddr \mapsto \aArgEval(\aexpr,\aenv,\astore)]
  \text.  
\end{align*}

\paragraph{Allocation, polyvariance and context-sensitivity}
\label{sec:polyvariance}
In the abstract semantics, the abstract allocation function
$\aalloc : \syn{Var} \times \sa{Conf} \to \sa{Addr}$ determines the
polyvariance of the analysis (and, by extension, its
context-sensitivity).
In a control-flow analysis, \emph{polyvariance} literally refers to
the number of abstract addresses (variants) there are for each
variable.
By selecting the right abstract allocation function, we can
instantiate pushdown versions of classical flow analyses.

\subparagraph{Monovariance: Pushdown 0CFA}

Pushdown 0CFA uses variables themselves for abstract addresses:

\begin{align*}
  \sa{Addr} &= \syn{Var}
  \\
  \alloc(v,\aconf) &= v
  \text.
\end{align*}

\subparagraph{Context-sensitive: Pushdown 1CFA}

Pushdown 1CFA pairs the variable with the current expression to get
an abstract address:

\begin{align*}
  \sa{Addr} &= \syn{Var} \times \syn{Exp}
  \\
  \alloc(\vv,(\expr,\aenv,\astore,\acont)) &= (\vv,\expr)
  \text.
\end{align*}

\subparagraph{Polymorphic splitting: Pushdown poly/CFA}

Assuming we compiled the program from a programming language with
let-bound polymorphism and marked which functions were let-bound, we
can enable polymorphic splitting:

\begin{align*}
  \sa{Addr} &= \syn{Var} + \syn{Var} \times \syn{Exp} 
  \\
  \alloc(\vv,(\sembr{\appform{\fexpr}{\aexpr}},\aenv,\astore,\acont)) &= 
  \begin{cases}
    (\vv,\sembr{\appform{\fexpr}{\aexpr}}) & \fexpr \text{ is let-bound}
    \\
    \vv & \text{otherwise}
    \text.
  \end{cases}
\end{align*}

\subparagraph{Pushdown $k$-CFA}

For pushdown $k$-CFA, we need to look beyond the current state
and at the last $k$ states.
By concatenating the expressions in the last $k$ states together, and
pairing this sequence with a variable we get pushdown $k$-CFA:
\begin{align*}
  \sa{Addr} &= \syn{Var} \times \syn{Exp}^k
  \\
  \aalloc(\vv,\vect{(\expr_1,\aenv_1,\astore_1,\acont_1),\ldots}) &=
  (\vv,\vect{\expr_1,\ldots,\expr_k})
  \text.
\end{align*}

\subsection{Partial orders}

For each set $\hat X$ inside the abstract configuration-space, we use the
natural partial order, $(\wt_{\hat X}) \subseteq \hat X \times \hat X$.
Abstract addresses and syntactic sets have flat partial orders.
For the other sets:
\begin{itemize}
    % Environments
  \item The partial order lifts point-wise over environments:
    \begin{equation*}
      \aenv \wt \aenv' \text{ iff }
      \aenv(\vv) = \aenv'(\vv) \text{ for all } \vv \in \dom(\aenv)
      \text.
    \end{equation*}

    % Closures
  \item The partial orders lifts component-wise over closures:
    \begin{equation*}
      (\lam,\aenv) \wt (\lam,\aenv') \text{ iff } \aenv \wt \aenv'
      \text.
    \end{equation*}
    
    % Stores
  \item The partial order lifts point-wise over stores:
    \begin{equation*}
      \astore \wt \astore' 
      \text{ iff }
      \astore(\aaddr) \wt \astore'(\aaddr) \text{ for all } \aaddr
      \in \dom(\astore)
      \text.
    \end{equation*}
    
    % Frames
  \item The partial order lifts component-wise over frames:
    \begin{equation*}
      (\vv,\expr,\aenv) \wt (\vv,\expr,\aenv') \text{ iff } \aenv \wt \aenv'
      \text.
    \end{equation*}
    
    %  Continuations
  \item The partial order lifts element-wise over continuations:
    \begin{equation*}
      \vect{\aphrame_1,\ldots,\aphrame_n} 
      \wt 
      \vect{\aphrame'_1,\ldots,\aphrame'_n} 
      \text{ iff }
      \aphrame_i \wt \aphrame'_i
      \text.
    \end{equation*}
    
    %  Configurations
  \item The partial order lifts component-wise across configurations:
    \begin{equation*}
      (\expr,\aenv,\astore,\acont) 
      \wt
      (\expr,\aenv',\astore',\acont')
      \text{ iff }
      \aenv \wt \aenv'
      \text{ and }
      \astore \wt \astore'
      \text{ and }
      \acont \wt \acont'
      \text.
    \end{equation*}
  \end{itemize}

% \todo{Order over sets of closures.}

\subsection{Soundness}

To prove soundness, we need an abstraction map $\absmap$ that connects the
concrete and abstract configuration-spaces:
\begin{align*}
  \absmap(\expr,\env,\store,\cont) &= (\expr,\absmap(\env),\absmap(\store),\absmap(\cont))
  \\
  \absmap(\env) &= \lt{\vv}{\absmap(\env(\vv))}
  \\
  \absmap(\store) &= \lt{\aaddr}{\!\!\! \bigjoin_{\absmap(\addr) = \aaddr} \!\!\! \set{\absmap(\store(\addr))}}
  \\
  \absmap\vect{\phrame_1,\ldots,\phrame_n} &= \vect{\absmap(\phrame_1), \ldots, \absmap(\phrame_n)}
  \\
  \absmap(\vv,\expr,\env) &= (\vv,\expr,\absmap(\env))
  \\
  \absmap(\addr) &\text{ is determined by the allocation functions}
  \text.
\end{align*}

It is easy to prove that the abstract transition relation
simulates the concrete transition relation:
\begin{theorem}
  If:
  \begin{equation*}
    \absmap(\conf) \wt \aconf \text{ and } \conf \To \conf'
    \text,
  \end{equation*}
  then there must exist $\aconf' \in \sa{Conf}$ such that:
  \begin{equation*}
    \absmap(\conf') \wt \aconf' \text{ and } \aconf \To \aconf'
    \text.
  \end{equation*}  
\end{theorem}
\begin{proof}
  The proof follows by case-wise analysis on the type of the
  expression in the configuration.
  It is a straightforward adaptation of similar proofs, such as that of
  \citet{mattmight:Might:2007:Dissertation} for $k$-CFA.
\end{proof}

\section{From the abstracted CESK machine
  to a PDA}
\label{sec:pda}
In the previous section, we constructed an infinite-state abstract
interpretation of the CESK machine.
The infinite-state nature of the abstraction makes it difficult to see how
to answer static analysis questions.
Consider, for instance, a control flow-question:
\begin{quote}
  At the call site $\appform{\fexpr}{\aexpr}$, may a closure over
  $\lam$ be called?
\end{quote}
If the abstracted CESK machine were a finite-state machine, an
algorithm could answer this question by enumerating all reachable
configurations and looking for an abstract configuration
$(\sembr{\appform{\fexpr}{\aexpr}},\aenv,\astore,\acont)$ in which
$(\lam,\_) \in \aArgEval(\fexpr,\aenv,\astore)$.
However, because the abstracted CESK machine may contain an infinite
number of reachable configurations, an algorithm cannot enumerate
them.

Fortunately, we can recast the abstracted CESK as a special kind of
infinite-state system: a pushdown automaton (PDA).
Pushdown automata occupy a sweet spot in the theory of computation:
they have an infinite configuration-space, yet many useful properties
(\eg{} word membership, non-emptiness, control-state
reachability) remain decidable.
Once the abstracted CESK machine becomes a PDA, we can answer the
control-flow question by checking whether a specific regular language,
when intersected with the language of the PDA, turns into  the empty 
language.

The recasting as a PDA is a shift in perspective.
A configuration has an expression, an environment and a store. 
A stack character is a frame.
% %
% A configuration may be decomposed into a finite set of control
% states---$\syn{Exp} \times \sa{Env} \times \sa{Store}$---and a stack
% over the alphabet $\sa{Frame}$.\footnote{
%   Not convinced the the set of control states is finite?
% %
%   Because abstract addresses are finite, and there are a finite number of variables in any given program, there are a finite number of abstract environments.
% %
%   A finite number of abstract environments, coupled with a finite number of \lamterm s for any given program implies a finite number of closures.
% %
%   A finite number of addresses and closures implies a finite number of abstract stores.
% %
%   All together, with a finite number of expressions, this implies a finite number of control states.
% }
%
We choose to make the alphabet the set of control states, so that the
language accepted by the PDA will be sequences of control-states
visited by the abstracted CESK machine.
Thus, every transition will consume the control-state to which it
transitioned as an input character.
Figure~\ref{fig:acesk-to-pda} defines the program-to-PDA conversion
function $\afPDA : \syn{Exp} \to \mathbb{PDA}$.  (Note the implicit
use of the isomorphism $\QState \times \sa{Kont} \cong \sa{Conf}$.)

\begin{figure}
%% \begin{small}
\begin{align*}
  \afPDA(\expr) &= (\QStates,\Alphabet,\StackAlpha,\transfunction,\qstate_0,\FStates,\vect{})
  \text{, where }
  \\
  \QStates &= \syn{Exp} \times \sa{Env} \times \sa{Store}
  \\
  \Alphabet &= \QStates
  \\
  \StackAlpha &= \sa{Frame}
  \\
  (\qstate,\epsilon,\qstate',\qstate') \in \transfunction
  & \text{ iff }
  (\qstate, \acont)
  \aTo
  (\qstate', \acont)
  \text{ for all } \acont
  \\
  (\qstate,\aphrame_{-},\qstate',\qstate') \in \transfunction
  & \text{ iff }
  (\qstate, \aphrame : \acont)
  \aTo
  (\qstate',\acont)
  \text{ for all } \acont
  \\
  (\qstate,\aphrame'_{+},\qstate',\qstate') \in \transfunction
  & \text{ iff }
  (\qstate, \acont)
  \aTo
  (\qstate',\aphrame' : \acont)
  \text{ for all } \acont
  \\
  (\qstate_0,\vect{}) &= \aInject(\expr)
  \\
  \FStates &= \QStates
  \text.
\end{align*}
\caption{$\afPDA : \syn{Exp} \to \mathbb{PDA}$.
}
\label{fig:acesk-to-pda}
%% \end{small}
\end{figure}

At this point, we can answer questions about whether a specified
control state is reachable by formulating a question about the
intersection of a regular language with a context-free language
described by the PDA.
That is, if we want to know whether the control state
$(\expr',\aenv,\astore)$ is reachable in a program $\expr$, we can reduce the problem to determining:
\begin{equation*}
  \Alphabet^*  \cdot \set{ (\expr',\aenv,\astore) } \cdot \Alphabet^* \mathrel{\cap} \Lang(\afPDA(\expr))
  \neq \emptyset
  \text,
\end{equation*}
where $L_1 \cdot L_2$ is the concatenation of formal languages $L_1$ and $L_2$.

\begin{theorem}
  Control-state reachability is decidable.
\end{theorem}
\begin{proof}
  The intersection of a regular language and a context-free language
  is context-free.
  The emptiness of a context-free language is decidable.
\end{proof}

Now, consider how to use control-state reachability to answer the
control-flow question from earlier.
There are a finite number of possible control states in which the
\lamterm{} $\lam$ may flow to the function $\fexpr$ in call site
$\appform{\fexpr}{\aexpr}$; let's call the this set of states $\hat
S$:
\begin{equation*}
  \hat S = 
  \setbuild{ (\sembr{\appform{\fexpr}{\aexpr}}, \aenv, \astore) }{
    (\lam,\aenv') \in \aArgEval(\fexpr,\aenv,\astore)
    \text{ for some } \aenv'
  }
  \text.
\end{equation*}
What we want to know is whether any state in the set $\hat S$ is
reachable in the PDA.
In effect what we are asking is whether there exists a control state $\qstate \in
\hat S$ such that:
\begin{align*}
  \Alphabet^* \cdot \set{\qstate} \cdot \Alphabet^*
  \mathrel{\cap}
  \Lang(\afPDA(\expr))
  \neq 
  \emptyset
  \text.
\end{align*}
If this is true, then $\lam$ may flow to $\fexpr$; if false, then it
does not.

%??? Move to front: Main problem.%
\paragraph{Problem: Doubly exponential complexity}
The non-emptiness-of-intersection approach establishes decidability of
pushdown control-flow analysis.
But, two exponential complexity barriers make this technique
impractical.

First, there are an exponential number of both environments
($\abs{\sa{Addr}}^{\abs{\syn{Var}}}$) and stores ($2^{\abs{\sa{Clo}}
  \times {\abs{\sa{Addr}}}}$) to consider for the set $\hat S$.
On top of that, computing the intersection of a regular language with
a context-free language will require enumeration of the
(exponential) control-state-space of the PDA.
% %
% To compute control-flow information for the whole program, this whole
% process must be conducted for every pairing between \lamterm s and
% call sites.
%
As a result, this approach is doubly exponential.
For the next few sections, our goal will be to lower the complexity of
pushdown control-flow analysis.

\section{Focusing on reachability}
\label{sec:pdreachability}

In the previous section, we saw that control-flow analysis reduces to
the reachability of certain control states within a pushdown system.
We also determined reachability by converting the
abstracted CESK machine into a PDA, and using emptiness-testing on a
language derived from that PDA.
Unfortunately, we also found that this approach is deeply exponential.

Since control-flow analysis reduced to the reachability of
control-states in the PDA, we skip the language problems and go
directly to reachability algorithms of
\citet{mattmight:Bouajjani:1997:PDA-Reachability,mattmight:Kodumal:2004:CFL,mattmight:Reps:1998:CFL}
and \citet{mattmight:Reps:2005:Weighted-PDA} that determine the
reachable \emph{configurations} within a pushdown system.
These algorithms are even polynomial-time.
Unfortunately, some of them are polynomial-time in the number of
control states, and in the abstracted CESK machine, there are an
exponential number of control states.
We don't want to \emph{enumerate} the entire control state-space, or
else the search becomes exponential in even the best case.

% \section{Focusing on efficiency: Dyck state graphs}
% \label{sec:dsg}
To avoid this worst-case behavior, we present a straightforward
pushdown-reachability algorithm that considers only the
\emph{reachable} control states.
We cast our reachability algorithm as a fixed-point iteration, in
which we incrementally construct the reachable subset of a pushdown system.
We term these algorithms ``iterative Dyck state graph construction.''

A \defterm{Dyck state graph} is a compacted, rooted pushdown
system
$G = (\DSStates,\DSFrames,\DSEdges,\dsstate_0)$, in which:
\begin{enumerate}
  \item $\DSStates$ is a finite set of nodes;
  \item $\DSFrames$ is a set of frames;
  \item $\DSEdges \subseteq \DSStates \times \DSFrames_\pm \times
    \DSStates$ is a set of stack-action edges; and
  \item $\dsstate_0$ is an initial state;
\end{enumerate}
such that for any node $\dsstate \in \DSStates$, it must be the case that:
\begin{equation*}
  (\dsstate_0,\vect{}) \mathrel{\PDTrans_G^*} (\dsstate,\vec{\dsframe})
  \text{ for some stack } \vec{\dsframe}
  \text.
\end{equation*}
In other words, a Dyck state graph is equivalent to a rooted pushdown
system in which there is a legal path to every control state from the
initial control state.\footnote{We chose the term \emph{Dyck state
    graph} because the sequences of stack actions along valid paths
  through the graph correspond to substrings in Dyck languages.
  A \defterm{Dyck language} is a language of balanced, ``colored''
  parentheses.
  In this case, each character in the stack alphabet is a color.}

We use $\mathbb{DSG}$ to denote the class of Dyck state graphs.
Clearly:
\begin{equation*}
  \mathbb{DSG} \subset \mathbb{RPDS}
  \text.
\end{equation*}
A Dyck state graph is a rooted pushdown system with the ``fat''
trimmed off; in this case, unreachable control states and unreachable
transitions are the ``fat.''

We can formalize the connection between rooted pushdown systems and
Dyck state graphs with a map:
\begin{equation*}
  \fDSG : \mathbb{RPDS} \to
  \mathbb{DSG}\text.
\end{equation*}
Given a rooted pushdown system $M =
(\ControlStates,\StackAlpha,\transfunction,\qstate_0)$, its equivalent Dyck state
graph is $\fDSG(M) = (\DSStates,\DSFrames,\DSEdges,\qstate_0)$, where the set $\DSStates$ contains reachable nodes:
\begin{equation*}
  \DSStates = \setbuild{ \qstate }{ (\qstate_0,\vect{}) \mathrel{\PDTrans_M^*} (\qstate,\vec{\stackchar}) \text{ for some stack } \vec{\stackchar} }
  \text,
\end{equation*}
and the set $\DSEdges$ contains reachable edges:
\begin{equation*}
  \DSEdges = \setbuild{ \qstate \pdedge^{\stackact} \qstate' }{
    \qstate
    \mathrel{\RPDTrans_M^{\stackact}}
    \qstate'
%
    % (\qstate_0, \vect{}) 
    %  \mathrel{\PDTrans_M^*}
    % (\qstate, \vec{\stackchar})
    % \text{ and }
    % (\qstate, \vec{\stackchar})
    % \mathrel{\PDTrans_M^{\stackact}}
    % (\qstate', \vec{\stackchar}')
  }
  \text,
\end{equation*}
and $\dsstate_0 = \qstate_0$.

In practice, the real difference between a rooted pushdown system and
a Dyck state graph is that our rooted pushdown system will be defined
intensionally (having come from the components of an abstracted CESK
machine), whereas the Dyck state graph will be defined extensionally,
with the contents of each component explicitly enumerated during its construction.

Our near-term goals are (1) to convert our abstracted CESK machine into
a rooted pushdown system and (2) to find an \emph{efficient} method
for computing an equivalent Dyck state graph from a rooted pushdown
system.

% \section{Converting an abstracted CESK machine into a rooted pushdown system}
% \label{sec:acesk-to-rpds}
% When we converted the abstracted CESK machine into a pushdown
% automaton, our goal was to decide for the reachability of certain
% control states.
% %
% We are shifting our focus to finding an efficient algorithm that
% enumerates the set of reachable control states.
% %
% Thus, a PDA is more heavyweight than we need: we won't be checking for
% language inclusion anymore.
%
% Since all we care about are the reachable control states of a PDA (and
% not the language it recognizes), we can convert an abstracted CESK
% machine more directly into a rooted pushdown system.
%
To convert the abstracted CESK machine into a rooted pushdown system,
we use the function $\afRPDS : \syn{Exp} \to \mathbb{RPDS}$:
\begin{align*}
  \afRPDS(\expr) &= (\QStates,\StackAlpha,\transfunction,\qstate_0)
  \\
  \QStates &= \syn{Exp} \times \sa{Env} \times \sa{Store}
  \\
  \StackAlpha &= \sa{Frame}
  \\
  %(\qstate,\epsilon,\qstate') 
  \qstate \pdedge^\epsilon \qstate'
  \in \transfunction
  & \text{ iff }
  (\qstate, \acont)
  \aTo
  (\qstate', \acont)
  \text{ for all } \acont
  \\
  \qstate \pdedge^{\aphrame_{-}} \qstate'
  %(\qstate,\aphrame_{-},\qstate')
  \in \transfunction
  & \text{ iff }
  (\qstate, \aphrame : \acont) 
  \aTo
  (\qstate',\acont)
  \text{ for all } \acont
  \\
  % (\qstate,\aphrame'_{+},\qstate') 
  \qstate \pdedge^{\aphrame_{+}} \qstate' 
  \in \transfunction
  & \text{ iff }
  (\qstate, \acont)
  \aTo
  (\qstate',\aphrame : \acont)
  \text{ for all } \acont
  \\
  (\qstate_0,\vect{}) &= \aInject(\expr)
  \text.
\end{align*}

\section{Compacting rooted pushdown systems}
\label{sec:rpds-to-dsg}
We now turn our attention to compacting a rooted pushdown
system (defined intensionally) into a Dyck state graph (defined
extensionally).
That is, we want to find an implementation of the function $\fDSG$.
To do so, we first phrase the Dyck state graph construction as the
least fixed point of a monotonic function.
This will provide a method (albeit an inefficient one) for computing
the function $\fDSG$.
In the next section, we look at an optimized work-list driven
algorithm that avoids the inefficiencies of this version.

The function $\mkDSG : \mathbb{RPDS} \to (\mathbb{DSG} \to
\mathbb{DSG})$ generates the monotonic iteration function we need:
\begin{align*}
  \mkDSG(M) &= f\text{, where }
  \\
  M &= (\QStates,\StackAlpha,\transfunction,\qstate_0)
  \\
  f(\DSStates,\DSFrames,\DSEdges,\dsstate_0) &= (\DSStates',\DSFrames,\DSEdges',\dsstate_0) \text{, where }
  \\
  \DSStates' &= \DSStates \union \setbuild{ \dsstate' }{ 
    \dsstate \in \DSStates 
    \text{ and }  
    \dsstate 
    \mathrel{\RPDTrans_M}
    \dsstate'
  }
  \union \set{\dsstate_0}
  \\
  \DSEdges' &= \DSEdges \union \setbuild{ \dsstate \pdedge^\stackact \dsstate' }{ 
    \dsstate \in \DSStates 
    \text{ and }  
    \dsstate 
    \mathrel{\RPDTrans_M^\stackact}
    \dsstate'    
  }
  \text.
\end{align*}
Given a rooted pushdown system $M$, each application of the function
$\mkDSG(M)$ accretes new edges at the frontier of the Dyck state
graph.
Once the algorithm reaches a fixed point, the Dyck state graph is complete:
\begin{theorem} 
  $\fDSG(M) = \lfp(\mkDSG(M))$.
\end{theorem}
\begin{proof}
  Let $M = (\QStates,\StackAlpha,\transfunction,\qstate_0)$.
  Let $f = \mkDSG(M)$.
  Observe that $\lfp(f) = f^n(\emptyset,\StackAlpha,\emptyset,\qstate_0)$ for some $n$.
  When $N \subseteq M$, then it easy to show that $f(N) \subseteq M$.
  Hence, $\fDSG(M) \supseteq \lfp(\mkDSG(M))$.

  To show $\fDSG(M) \subseteq \lfp(\mkDSG(M))$, suppose this is not
  the case.
  Then, there must be at least one edge in $\fDSG(M)$ that is not in
  $\lfp(\mkDSG(M))$.
  Let $(\dsstate,\stackact,\dsstate')$ be one such edge, such that the state $\dsstate$ \emph{ is } in 
  $\lfp(\mkDSG(M))$.
  Let $m$ be the lowest natural number such that $\dsstate$ appears in $f^m(M)$.
  By the definition of $f$, this edge must appear in $f^{m+1}(M)$, which means it must also appear in 
  $\lfp(\mkDSG(M))$, which is a contradiction.
  Hence, $\fDSG(M) \subseteq \lfp(\mkDSG(M))$.
\end{proof}

\subsection{Complexity: Polynomial and exponential}
 
To determine the complexity of this algorithm, we ask two questions:
how many times would the algorithm invoke the iteration function in
the worst case, and how much does each invocation cost in the
worst case?
The size of the final Dyck state
graph bounds the run-time of the algorithm.
Suppose the final Dyck state graph has $m$ states.
In the worst case, the iteration function adds only a single edge each
time.
Since there are at most $2\abs{\StackAlpha}m^2 + m^2$ edges in the final graph, the maximum
number of iterations is $2\abs{\StackAlpha}m^2 + m^2$.

The cost of computing each iteration is harder to bound.
The cost of determining whether to add a push edge is proportional to
the size of the stack alphabet, while the cost of determining whether
to add an $\epsilon$-edge is constant, so the cost of determining all
new push and pop edges to add is proportional to $\abs{\StackAlpha}m +
m$.
Determining whether or not to add a pop edge is expensive.
To add the pop edge
$\triedge{\dsstate}{\stackchar_-}{\dsstate'}$, we must prove that
there exists a configuration-path to the control state $\dsstate$, in which
the character $\stackchar$ is on the top of the stack.
This reduces to a CFL-reachability query~\cite{mattmight:Melski:2000:CFL}
at each node, the cost of which is $O(\abs{\StackAlpha_\pm}^3
m^3)$~\cite{mattmight:Kodumal:2004:CFL}.

To summarize, in terms of the number of reachable control states, the
complexity of the most recent algorithm is:
\[
O((2\abs{\StackAlpha}m^2 + m^2) 
  \times
  (\abs{\StackAlpha}m + m
  + \abs{\StackAlpha_\pm}^3m^3)) 
  =  O(\abs{\StackAlpha}^4m^5)\text.
\]
While this approach is polynomial in the number of reachable
control states, it is far from efficient.
In the next section, we provide an optimized version of this
fixed-point algorithm that maintains a work-list and an
\ecg{} to avoid spurious recomputation.

% Moreover, we have carefully phrased the complexity in terms of
% ``reachable'' control states because we expect that, in practice, Dyck
% state graphs will be extremely sparse, and because, the maximum number
% of control states is exponential in the size of the input program.
% %
% After the subsequent refinement, we will be able to develop a hierarchy of
% pushdown control-flow analyses that employs widening to achieve a
% polynomial-time algorithm at its foundation.

\section{Efficiency: Work-lists and 
    \texorpdfstring{$\boldsymbol\epsilon$}{epsilon}-closure graphs}
\label{sec:ecg-worklist}
We have developed a fixed-point formulation of the Dyck state graph
construction algorithm, but found that, in each iteration, it wasted
effort by passing over all discovered states and edges, even though
most will not contribute new states or edges.
Taking a cue from graph search, we can adapt the fixed-point algorithm
with a work-list.
That is, our next algorithm will keep a work-list of new states and
edges to consider, instead of reconsidering all of them.
In each iteration, it will pull new states and edges from the work
list, insert them into the Dyck state graph and then populate the
work-list with new states and edges that have to be added as a
consequence of the recent additions.

\subsection{\texorpdfstring{$\boldsymbol\epsilon$}{epsilon}-closure graphs}
Figuring out what edges to add as a consequence of another edge
requires care, for adding an edge can have ramifications on distant
control states.
Consider, for example, adding the $\epsilon$-edge $\qstate
\pdedge^{\epsilon} \qstate'$ into the following graph:
\begin{equation*}
  \xymatrix{ 
    \qstate_0 \ar[r]^{\stackchar_+} & \qstate & \qstate' \ar[r]^{\stackchar_-} & \qstate_1
  } 
\end{equation*}
As soon this edge drops in, an $\epsilon$-edge ``implicitly'' appears
 between $\qstate_0$ and $\qstate_1$ because the net
stack change between them is empty; the resulting graph looks like:
\begin{equation*}
  \xymatrix{ 
    \qstate_0 \ar[r]^{\stackchar_+} \ar @{-->} @(u,u) [rrr]^\epsilon &
    \qstate \ar[r]^\epsilon &
    \qstate' \ar[r]^{\stackchar_-} &
    \qstate_1
  }  
\end{equation*}
where we have illustrated the implicit $\epsilon$-edge as a dotted line.

To keep track of these implicit edges, we will construct a second
graph in conjunction with the Dyck state graph: an $\epsilon$-closure
graph.
In the $\epsilon$-closure graph, every edge indicates the existence of
a no-net-stack-change path between control states.
The $\epsilon$-closure graph simplifies the task of figuring out
which states and edges are impacted by the addition of a new edge.

Formally, an \textbf{$\boldsymbol \epsilon$-closure graph}, is a pair
$G_\epsilon = (N,H)$, where $N$ is a set of states, and $H \subseteq N \times N$ is
a set of edges.
Of course, all \ecg s are reflexive: every node has a self loop.
We use the symbol $\mathbb{ECG}$ to denote the class of all \ecg s.

We have two notations for finding ancestors and descendants of a state
in an \ecg{} $G_\epsilon = (N,H)$:
\begin{align*}
  \eancestor{\dsstate} &= \setbuild{ \dsstate' }{ (\dsstate',\dsstate) \in H } 
  && \text{[ancestors]}
  \\
  \edescendent{\dsstate} &= \setbuild{ \dsstate' }{ (\dsstate,\dsstate') \in H } 
  && \text{[descendants]}
  \text.
\end{align*}

\subsection{Integrating a work-list}
Since we only want to consider new states and edges in each iteration,
we need a work-list, or in this case, two work-graphs.
A Dyck state work-graph is a pair $(\Delta S,\Delta E)$ in which the
set $\Delta S$ contains a set of states to add, and the set $\Delta E$
contains edges to be added to a Dyck state
graph.\footnote{Technically, a work-graph is not an actual graph,
  since $\Delta E \not\subseteq \Delta S \times \StackAlpha_\pm \times
  \Delta S$; a work-graph is just a set of nodes and a set of edges.}
We use $\Delta\mathbb{DSG}$ to refer to the class of all Dyck state work-graphs.

An $\epsilon$-closure work-graph is a set $\Delta H$ of new
$\epsilon$-edges.
We use $\Delta\mathbb{ECG}$ to refer to the class of all
$\epsilon$-closure work-graphs.

\subsection{A new fixed-point iteration-space}
Instead of consuming a Dyck state graph and producing a Dyck state
graph, the new fixed-point iteration function will consume and produce a Dyck state graph,
an \ecg{}, a Dyck state work-graph and an $\epsilon$-closure work graph.
Hence, the iteration space of the new algorithm is:
\begin{equation*}
  \s{IDSG} = \mathbb{DSG} \times \mathbb{ECG}
  \times \Delta\mathbb{DSG}
  \times \Delta\mathbb{ECG}
  \text.
\end{equation*}
(The \emph{I} in $\s{IDSG}$ stands for \emph{intermediate}.)

\subsection{The \texorpdfstring{$\boldsymbol\epsilon$}{epsilon}-closure 
  graph work-list algorithm}
The function $\mkDSG' : \mathbb{RPDS} \to (\s{IDSG} \to
\s{IDSG})$ generates the required iteration function (Figure~\ref{fig:mkdsg-ecg}).
\begin{figure}
%% \begin{small}
\begin{align*}
  \mkDSG'(M)
  &= 
  f\text{, where }
  \\
  M &= (\QStates,\StackAlpha,\transfunction,\qstate_0)
  \\
  f(G,G_\epsilon,\Delta G,\Delta H) &=
   (G',G'_\epsilon,\Delta G',\Delta H' -  H)\text{, where }
  \\
  (\DSStates,\DSFrames,\DSEdges,\dsstate_0)  &= G 
  \\
  (S,H)  &= G_\epsilon
  \\
  (\Delta S,\Delta E) &= \Delta G
  \\
  (\Delta E_0, \Delta H_0) &= \!\! \Union_{\dsstate \in \Delta S} \!\!  \fsprout_M(\dsstate)
  \\
  (\Delta E_1, \Delta H_1) &= 
  \!\!\!\!\!\!\!\!\!\!
  \Union_{(\dsstate,\stackchar_+,\dsstate') \in \Delta E} 
  \!\!\!\!\!\!\!\!\!\!
  \faddpush_M(G,G_\epsilon) (\dsstate,\stackchar_+,\dsstate') 
  \\
  (\Delta E_2, \Delta H_2) &=
  \!\!\!\!\!\!\!\!\!\!
  \Union_{(\dsstate,\stackchar_-,\dsstate') \in \Delta E} 
  \!\!\!\!\!\!\!\!\!\!
  \faddpop_M(G,G_\epsilon) (\dsstate,\stackchar_-,\dsstate')  
  \\
  (\Delta E_3, \Delta H_3) &=
  \!\!\!\!\!\!\!\!
  \Union_{(\dsstate,\epsilon,\dsstate') \in \Delta E} 
  \!\!\!\!\!\!\!\!
  \faddempty_M(G,G_\epsilon) (\dsstate,\dsstate')   
  \\
  (\Delta E_4, \Delta H_4) &=
  \!\!\!\!\!\!
  \Union_{(\dsstate,\dsstate') \in \Delta H} 
  \!\!\!\!\!\!\!
  \faddempty_M(G,G_\epsilon) (\dsstate,\dsstate')   
  \\
  S' &= \DSStates \union \Delta S
  \\
  E' &= \DSEdges \union \Delta E
  \\
  H' &= H \union \Delta H
  \\
  \Delta E' &= \Delta E_0 \union \Delta E_1 \union \Delta E_2 \union \Delta E_3 \union \Delta E_4
  \\
  \Delta S' &= \setbuild{ \dsstate' }{ (\dsstate,\stackact,\dsstate') \in \Delta E' }
  \\
  \Delta H' &= \Delta H_0 \union \Delta H_1 \union \Delta H_2 \union \Delta H_3 \union \Delta H_4
  \\
  G' &= (\DSStates \union \Delta S,\DSFrames,E',\qstate_0)
  \\
  G_\epsilon' &= (S',H')
  \\
  \Delta G' &= (\Delta S' - S', \Delta E ' - E')
  \text.
\end{align*}
\caption{The fixed point of the function $\mkDSG'(M)$ contains
  the Dyck state graph of the rooted pushdown system $M$.}
\label{fig:mkdsg-ecg}
%% \end{small}
\end{figure}
Please note that we implicitly distribute union across tuples:
\begin{equation*}
  (X,Y) \union (X',Y') = 
  (X \union X, Y \union Y')
  \text.
\end{equation*}
The functions $\fsprout$, $\faddpush$, $\faddpop$, $\faddempty$
calculate the additional the Dyck state graph edges and \ecg{} edges
(potentially) introduced by a new state or edge.

\paragraph{Sprouting}
Whenever a new state gets added to the Dyck state graph, 
the algorithm must check whether that state has any new edges to contribute.
Both push edges and $\epsilon$-edges do not depend on the current
stack, so any such edges for a state in the pushdown system's
transition function belong in the Dyck state graph. 
The sprout function:
\begin{equation*}
  \fsprout_{(\QStates,\StackAlpha,\transfunction)} :
\QStates \to (\Pow{\transfunction} \times \Pow{\QStates \times \QStates})
\text,
\end{equation*}
checks whether a new state could produce any new push edges or no-change edges.
We can represent its behavior diagrammatically:
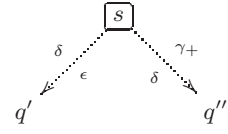
\begin{equation*}
  \xymatrix{
    &
    *+[F-:<1pt>]{\dsstate} \ar @{..>} [dl]^{\epsilon}_\transfunction \ar @{..>} [dr]^{\stackchar_+}_\transfunction  & 
    \\
     \qstate' && \qstate''
   }
\end{equation*}
which means if adding control state $\dsstate$:
\begin{itemize}
\item[] add edge $\triedge{\dsstate}{\epsilon}{\qstate'}$ if it exists in $\delta$, and
\item[] add edge $\triedge{\dsstate}{\stackchar_+}{\qstate''}$ if it exists in  $\delta$.
\end{itemize}
Formally:
\begin{align*}
  \fsprout_{(\QStates,\StackAlpha,\transfunction)} (\dsstate) &= (\Delta E, \Delta H)\text{, where }
\end{align*}
\begin{align*}
  \Delta E &= \setbuild{ \triedge{\dsstate}{\epsilon}{\qstate} }{ \triedge{\dsstate}{\epsilon}{\qstate} \in \transfunction }
  \mathrel{\union}
   \setbuild{ \triedge{\dsstate}{\stackchar_+}{\qstate} }{ \triedge{\dsstate}{\stackchar_+}{\qstate} \in \transfunction }
  \\
  \Delta H &= \setbuild{ \biedge{\dsstate}{\qstate} }{ \triedge{\dsstate}{\epsilon}{\qstate} \in \transfunction } 
  \text.
\end{align*}

\paragraph{Considering the consequences of a new push edge}
Once our algorithm adds a new push edge to a Dyck state graph, there
is a chance that it will enable new pop edges for the same stack frame
somewhere downstream.
If and when it does enable pops, it will also add new edges to the
\ecg{}.
The $\faddpush$ function:
\begin{equation*}
  \faddpush_{(\QStates,\StackAlpha,\transfunction)} :
\mathbb{DSG} \times \mathbb{ECG} \to \transfunction \to (\Pow{\transfunction} \times \Pow{\QStates \times \QStates})
\text,
\end{equation*}
checks for $\epsilon$-reachable states that could produce a pop.
We can represent this action diagrammatically:
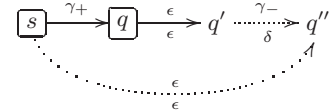
\begin{equation*}
  \xymatrix{
    *+[F-:<1pt>]{\dsstate} \ar [r]^{\stackchar_+} \ar @(d,d) @{..>} [rrr]^\epsilon_\epsilon
    & 
    *+[F-:<1pt>]{\qstate} \ar [r]^{\epsilon}_\epsilon
    & 
    {\qstate'} \ar @{..>} [r]^{\stackchar_-}_\transfunction
    & \qstate''
    \\
    \\
 }
\end{equation*}
which means if adding push-edge $\triedge{\dsstate}{\stackchar_+}{\qstate}$:
\begin{itemize}
  \item[] if pop-edge $\triedge{\qstate'}{\stackchar_-}{\qstate''}$ is in $\delta$, then
  \item[] \hspace{1.5em} add edge $\triedge{\qstate'}{\stackchar_-}{\qstate''}$, and
  \item[] \hspace{1.5em} add $\epsilon$-edge $\biedge{\dsstate}{\qstate''}$.
\end{itemize}
Formally:
\begin{align*}
  &\faddpush_{(\QStates,\StackAlpha,\transfunction)} (G,G_\epsilon) (\triedge{\dsstate}{\stackchar_+}{\qstate}) =
  (\Delta E, \Delta H)\text{, where }
  \\
  &\;\;\;\Delta E = \setbuild{ \triedge{\qstate'}{\stackchar_-}{\qstate''} }{
    \qstate' \in \edescendent{\qstate}
    \text{ and }
    \triedge{\qstate'}{\stackchar_-}{\qstate''} \in \transfunction }
  \\
  &\;\;\; \Delta H = \setbuild{ \biedge{\dsstate}{\qstate''} }{ 
    \qstate' \in \edescendent{\qstate}
    \text{ and }
    \triedge{\qstate'}{\stackchar_-}{\qstate''} \in \transfunction }
  \text.
\end{align*}

\paragraph{Considering the consequences of a new pop edge}
%
% TODO: Consider eliminating \Delta E return.
%
Once the algorithm adds a new pop edge to a Dyck state graph, it will create 
at least one new \ecg{} edge and possibly more by matching up with
upstream pushes.
The $\faddpop$ function:
\begin{equation*}
  \faddpop_{(\QStates,\StackAlpha,\transfunction)} :
\mathbb{DSG} \times \mathbb{ECG} \to \transfunction \to (\Pow{\transfunction} \times \Pow{\QStates \times \QStates})
\text,
\end{equation*}
checks for $\epsilon$-reachable push-edges that could match this pop-edge.
We can represent this action diagrammatically:
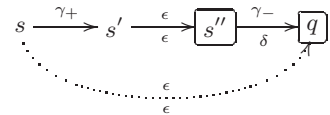
\begin{equation*}
  \xymatrix{
    {\dsstate} \ar [r]^{\stackchar_+} \ar @(d,d) @{..>} [rrr]^\epsilon_\epsilon
    & 
    {\dsstate'} \ar [r]^{\epsilon}_\epsilon
    & 
    *+[F-:<1pt>]{\dsstate''} \ar [r]^{\stackchar_-}_\transfunction
    & 
    *+[F-:<1pt>]{\qstate}
    \\
    \\
 }
\end{equation*}
which means if adding pop-edge $\triedge{\dsstate''}{\stackchar_-}{\qstate}$:
\begin{itemize}
 \item[] if push-edge $\triedge{\dsstate}{\stackchar_+}{\dsstate'}$ is already in the Dyck state graph, then
 \item[] \hspace{1.5em} add $\epsilon$-edge $\biedge{\dsstate}{\qstate}$.
\end{itemize}
Formally:
\begin{align*}
  &\faddpop_{(\QStates,\StackAlpha,\transfunction)} (G,G_\epsilon) (\triedge{\dsstate''}{\stackchar_-}{\qstate}) =
  (\Delta E, \Delta H)\text{, where }
\\
&  \Delta E = \emptyset\mbox{ and }
  \Delta H = \setbuild{ \biedge{\dsstate}{\qstate} }{ 
    \dsstate' \in \eancestor{\dsstate''}
    \text{ and }
    \triedge{\dsstate}{\stackchar_+}{\dsstate'} \in G }
  \text.
\end{align*}

\paragraph{Considering the consequences of a new $\boldsymbol\epsilon$-edge}
%
% TODO: Consider eliminating \Delta E return.
%
Once the algorithm adds a new \ecg{} edge, it may transitively
have to add more \ecg{} edges, and it may connect an old push to
(perhaps newly enabled) pop edges.
The $\faddempty$ function:
\[
\begin{array}{l}
  \faddempty_{(\QStates,\StackAlpha,\transfunction)} :
\\
\qquad
\mathbb{DSG} \times \mathbb{ECG} \to (\QState \times \QState) \to (\Pow{\transfunction} \times \Pow{\QStates \times \QStates})
\text,
\end{array}
\]
checks for newly enabled pops and \ecg{} edges:
Once again, we can represent this action diagrammatically:
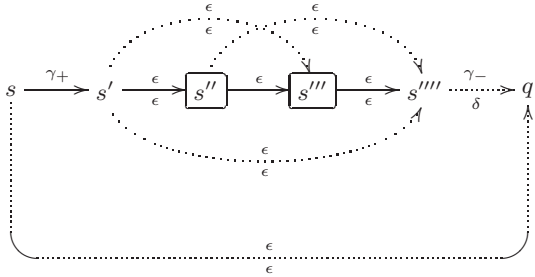
\begin{equation*}
  \xymatrix{
    \\
   {\dsstate} \ar [r]^{\stackchar_+} 
   % \ar @<-10pt> @(d,d) @{..>} [rrrrr]^\epsilon_\epsilon
   \ar @{..>} `d[ddr] `r[rrrrr]^\epsilon_\epsilon [rrrrr]
   &
   {\dsstate'} \ar [r]^\epsilon_\epsilon \ar @(u,u) @{..>} [rr]^\epsilon_\epsilon
   \ar @(d,d) @{..>} [rrr]^\epsilon_\epsilon
   & 
   *+[F-:<1pt>]{\dsstate''} \ar [r]^{\epsilon} \ar @(u,u) @{..>} [rr]^\epsilon_\epsilon
   &
   *+[F-:<1pt>]{\dsstate'''} \ar [r]^\epsilon_\epsilon
   & 
   {\dsstate''''} \ar @{..>} [r]^{\stackchar_-}_\transfunction
   & 
   {\qstate}
   \\
   & & & & \\
   & & & & \\
}
\end{equation*}
which means if adding $\epsilon$-edge $\biedge{\dsstate''}{\dsstate'''}$:
\begin{itemize}
  \item[] if pop-edge $\triedge{\dsstate''''}{\stackchar_-}{\qstate}$ is in $\transition$, then
  \item[] \hspace{2em} add $\epsilon$-edge $\biedge{\dsstate}{\qstate}$; and
  \item[] \hspace{2em} add edge $\triedge{\dsstate''''}{\stackchar_-}{\qstate}$; 
  \item[] add $\epsilon$-edges $\biedge{\dsstate'}{\dsstate'''}$,
    $\biedge{\dsstate''}{\dsstate''''}$, and
    $\biedge{\dsstate'}{\dsstate''''}$.
\end{itemize}
Formally:
\begin{align*}
  &\faddempty_{(\QStates,\StackAlpha,\transfunction)} (G,G_\epsilon) (\biedge{\dsstate''}{\dsstate'''})
  = (\Delta E, \Delta H)\text{, where }
  \\
  &\;\;\;\;\Delta E = \big\{ \triedge{\dsstate''''}{\stackchar_-}{\qstate} :
    \dsstate' \in \eancestor{\dsstate''}
    \text{ and }
    \dsstate'''' \in \edescendent{\dsstate'''}
    \text{ and } \\
    & \hspace{9.5em} \triedge{\dsstate}{\stackchar_+}{\dsstate'} \in G \big\}
  \\
  &\;\;\;\;\Delta H = \big\{ \biedge{\dsstate}{\qstate} :
    \dsstate' \in \eancestor{\dsstate''}
    \text{ and }
    \dsstate'''' \in \edescendent{\dsstate'''}
    \text{ and } \\
    & \hspace{8.5em} \triedge{\dsstate}{\stackchar_+}{\dsstate'} \in G \big\}
  \\
  & \hspace{3em} \union \setbuild{ \biedge{\dsstate'}{\dsstate'''}  }{
    \dsstate' \in \eancestor{\dsstate''}
  }
  \\
  & \hspace{3em} \union \setbuild{ \biedge{\dsstate''}{\dsstate''''}  }{
    \dsstate'''' \in \edescendent{\dsstate'''}
   }
  \\
  & \hspace{3em} \union \setbuild{ \biedge{\dsstate'}{\dsstate''''}  }{
    \dsstate' \in \eancestor{\dsstate''}
    \text{ and }
    \dsstate'''' \in \edescendent{\dsstate'''}
  }
  \text.
\end{align*}

\subsection{Termination and correctness}
Because the iteration function is no longer monotonic, we have to
prove that a fixed point exists.
It is trivial to show that the Dyck state graph component of the
iteration-space ascends monotonically with each application; that is:
\begin{lemma}
  Given $M \in \mathbb{RPDS}, G \in \mathbb{DSG}$ such that $G \subseteq M$, 
  if $\mkDSG'(M)(G,G_\epsilon,\Delta G) = (G',G_\epsilon',\Delta G')$, then
  $G \subseteq G'$.
\end{lemma}
Since the size of the Dyck state graph is bounded by the original
pushdown system $M$, the Dyck state graph will eventually reach a
fixed point.
Once the Dyck state graph reaches a fixed point, both work-graphs/sets
will be empty, and the \ecg{} will also stabilize.
We can also show that this algorithm is correct:
\begin{theorem}
  $\lfp(\mkDSG'(M)) = (\fDSG(M),G_\epsilon,(\emptyset,\emptyset),\emptyset)$.
\end{theorem}
\begin{proof}
  The proof is similar in structure to the previous one.
\end{proof}

\subsection{Complexity: Still exponential, but more efficient}

As with the previous algorithm, to determine the complexity of this
algorithm, we ask two questions: how many times would the algorithm
invoke the iteration function in the worst case, and how much
does each invocation cost in the worst case?
The run-time of the algorithm is bounded by the size of the final Dyck state
graph plus the size of the \ecg.
Suppose the final Dyck state graph has $m$ states.
In the worst case, the iteration function adds only a single edge each
time.
There are at most $2\abs{\StackAlpha}m^2 + m^2$ edges in the Dyck
state graph and at most $m^2$ edges in the \ecg{}, which bounds the
number of iterations.

Next, we must reason about the worst-case cost of adding an edge: how
many edges might an individual iteration consider?
In the worst case, the algorithm will consider every edge in every
iteration, leading to an asymptotic time-complexity of:
\begin{equation*}
  O((2\abs{\StackAlpha}m^2 + 2m^2)^2) = 
  O(\abs{\StackAlpha}^2m^4)
  \text.
\end{equation*}
While still high, this is a an improvement upon the previous
algorithm.  
%
% In practice, Dyck state graphs are likely to be sparse, with only
% a few edges considered at each iteration, so that the average-case
% complexity is closer to $O(\abs{\StackAlpha}m^2)$.
%
For sparse Dyck state graphs, this is a reasonable algorithm.

\section{Polynomial-time complexity from widening}
\label{sec:widening}

In the previous section, we developed a more efficient fixed-point
algorithm for computing a Dyck state graph.
Even with the core improvements we made, the algorithm remained
exponential in the worst case, owing to the fact that there could be
an exponential number of reachable control states.
When an abstract interpretation is intolerably complex, the standard
approach for reducing complexity and accelerating convergence is
widening~\cite{mattmight:Cousot:1977:AI}.
(Of course, widening techniques trade away some precision to gain this
speed.)
It turns out that the small-step variants of finite-state CFAs are
exponential without some sort of widening as well.

To achieve polynomial time complexity for pushdown control-flow
analysis requires the same two steps as the classical case: (1)
widening the abstract interpretation to use a global,
``single-threaded'' store and (2) selecting a monovariant allocation
function to collapse the abstract configuration-space.
Widening eliminates a source of exponentiality in the size of the
store; monovariance eliminates a source of exponentiality from
environments.
In this section, we redevelop the pushdown control-flow analysis
framework with a single-threaded store and calculate
its complexity.

\subsection{Step 1: Refactor the concrete semantics}
First, consider defining the reachable states of the concrete
semantics using fixed points.
That is, let the system-space of the evaluation function be
sets of configurations:
\begin{equation*}
  \system \in \s{System} = 
  \Pow{\s{Conf}} =
  \Pow{\syn{Exp} \times \s{Env} \times \s{Store} \times \s{Kont}}
  \text.
\end{equation*}
We can redefine the concrete evaluation function:
\begin{align*}
  \Eval(\expr) &= \lfp(\tf_\expr)\text{, where }
  \tf_\expr: \s{System} \to \s{System}
  \text{ and }
  \\
  \tf_\expr(\system) &= \set{\Inject(\expr)} \union \setbuild{ \conf' }{ \conf \in \system \text{ and } \conf \To \conf' }
  \text.
\end{align*}

\subsection{Step 2: Refactor the abstract semantics}
We can take the same approach with the abstract evaluation function,
first redefining the abstract system-space:
\begin{align*}
  \asystem  \in \sa{System} &= 
  \Pow{\sa{Conf}} 
\\
&=
  \Pow{\syn{Exp} \times \sa{Env} \times \sa{Store} \times \sa{Kont}}
  \text,
\end{align*}
and then the abstract evaluation function:
\begin{align*}
  \aEval(\expr) &= \lfp(\atf_\expr)\text{, where } \atf_\expr : \sa{System} \to \sa{System} \text{ and }
  \\
  \atf_\expr(\asystem) &= \set{\aInject(\expr)} \union \setbuild{ \aconf' }
  { \aconf \in \asystem \text{ and } \aconf \aTo \aconf' }
  \text.
\end{align*}
What we'd like to do is shrink the abstract system-space with a
refactoring that corresponds to a widening.

\subsection{Step 3: Single-thread the abstract store}
We can approximate a set of abstract stores
$\set{\astore_1,\ldots,\astore_n}$ with 
the least-upper-bound of those stores: $\astore_1 \join \cdots
\join \astore_n$.
We can exploit this by creating a new abstract system space in which
the store is factored out of every configuration.
Thus, the system-space contains a set of \emph{partial configurations}
and a single global store:
\begin{align*}
  \sa{System}' &= \Pow{\sa{PConf}} \times \sa{Store}
  \\
  \apconf \in \sa{PConf} &= \syn{Exp} \times \sa{Env} \times \sa{Kont}
  \text.
\end{align*}
We can factor the store out of the abstract transition relation as well, so that
$(\afTo^\astore) \subseteq \sa{PConf} \times (\sa{PConf} \times \sa{Store})$:
\begin{align*}
  (\expr, \aenv, \acont) \mathrel{\afTo^{\astore}} ((\expr',\aenv',\acont'), \astore') 
  \text{ iff }
(\expr, \aenv, \astore, \acont) \aTo (\expr',\aenv', \astore', \acont') 
\text,
\end{align*}
which gives us a new iteration function,
$\atf'_\expr : \sa{System}' \to \sa{System}'$,
\begin{align*}
  \atf'_\expr(\hat P,\astore) &= (\hat P',\astore')\text{, where } 
  \\
  \hat P' &= \setbuild{ \apconf' }{ \apconf \mathrel{\afTo^\astore} (\apconf',\astore'') } \union \set{ \apconf_0 } 
  \\
  \astore' &= \bigjoin \setbuild{ \astore'' }{ \apconf \mathrel{\afTo^\astore} (\apconf',\astore'') }
  \\
  (\apconf_0,\vect{}) &= \aInject(\expr)
  \text.
\end{align*}

\subsection{Step 4: Dyck state control-flow graphs}
Following the earlier Dyck state graph reformulation of the
pushdown system, we can reformulate the set of partial
configurations as a \emph{Dyck state control-flow graph}.
A \defterm{Dyck state control-flow graph} is a frame-action-labeled
graph over partial control states, and a \defterm{partial control
  state} is an expression paired with an environment:
\begin{align*}
  \sa{System}'' &= \sa{DSCFG} \times \sa{Store}
  \\
  \sa{DSCFG} &= \PowSm{\sa{PState}} \times \PowSm{\sa{PState} \times \sa{Frame}_\pm \times \sa{PState}}
  \\
  \apstate \in \sa{PState} &= \syn{Exp} \times \sa{Env}
  \text.
\end{align*}
In a Dyck state control-flow graph, the partial control states are
partial configurations which have dropped the continuation component;
the continuations are encoded as paths through the graph.

If we wanted to do so, we could define a new monotonic iteration function
analogous to the simple fixed-point formulation of
Section~\ref{sec:rpds-to-dsg}:
\begin{equation*}
  \atf_\expr : \sa{System}'' \to \sa{System}''
  \text,
\end{equation*}
again using CFL-reachability to add pop edges at each step.

\paragraph{A preliminary analysis of complexity}
Even without defining the system-space iteration function, we can ask,
\emph{How many iterations will it take to reach a fixed point in the worst
case?}
This question is really asking, \emph{How many edges can we add?}
And, \emph{How many entries are there in the store?}
Summing these together, we arrive at the worst-case number of
iterations:
\begin{align*}
  \overbrace{\abs{\sa{PState}}
  \times
  \abs{\sa{Frame}_\pm}
  \times
  \abs{\sa{PState}}}^{\text{DSCFG edges}}
  +
  \overbrace{
    \abs{\sa{Addr}}
    \times
    \abs{\sa{Clo}}
  }^{\text{store entries}}
  \text.
\end{align*}
With a monovariant allocation scheme that eliminates abstract environments, the number of iterations
ultimately reduces to:
\begin{align*}
  \abs{\syn{Exp}}
  \times (2 \abs{\sa{\syn{Var}}} + 1)
  \times \abs{\syn{Exp}}
  +
  \abs{\syn{Var}}
  \times \abs{\syn{Lam}}
  \text,
\end{align*}
which means that, in the worst case, the algorithm makes a cubic
number of iterations with respect to the size of the input
program.\footnote{ In computing the number of frames, we note that in
  every continuation, the variable and the expression uniquely
  determine each other based on the let-expression from which they
  both came.
  As a result, the number of abstract frames available in a
  monovariant analysis is bounded by both the number of variables and
  the number of expressions, \ie, $\abs{\sa{Frame}} =
  \abs{\syn{Var}}$.}

The worst-case cost of the each iteration would be dominated by a
CFL-reachability calculation, which, in the worst case, must consider
every state and every edge:
\begin{equation*}
  O(\abs{\syn{Var}}^3 \times \abs{\syn{Exp}}^3)
  \text.
\end{equation*}
Thus, each iteration takes $O(n^{6})$ and there are a maximum of $O(n^{3})$ iterations, where $n$ is the size of the program.
So,total complexity would be $O(n^{9})$ for a monovariant
pushdown control-flow analysis with this scheme, where $n$ is again the size of the program.
Although this algorithm is polynomial-time, we can do better.

\subsection{Step 5: Reintroduce 
  \texorpdfstring{$\boldsymbol\epsilon$}{epsilon}-closure graphs}
Replicating the evolution from Section~\ref{sec:ecg-worklist} for this
store-widened analysis, we arrive at a more efficient polynomial-time
analysis.
An \ecg{} in this setting is a set of pairs of store-less,
continuation-less partial states:
\begin{align*}
  \sa{ECG} &= \Pow{\sa{PState} \times \sa{PState}}
  \text.
\end{align*}
Then, we can set the system space to include \ecg s:
\begin{align*}
  \sa{System}''' &= \sa{DSG} \times \sa{ECG} \times \sa{Store}\text.
\end{align*}

Before we redefine the iteration function, we need another factored
transition relation.
The stack- and action-factored transition relation
$(\apTo^{\astore}_{\stackact}) \subseteq \sa{PState} \times
\sa{PState} \times \s{Store}$ determines if a transition is possible
under the specified store and stack-action:
\begin{align*}
  \pdcfato{(\expr,\aenv)}{\astore}{\aphrame_+}{((\expr',\aenv'),\astore')}
  & \text{ iff }
  (\expr,\aenv,\astore,\acont) 
  \aTo
  (\expr',\aenv',\astore',\aphrame : \acont')
  \\
  \pdcfato{(\expr,\aenv)}{\astore}{\aphrame_-}{((\expr',\aenv'),\astore')}
  & \text{ iff }
  (\expr,\aenv,\astore,\aphrame : \acont) 
  \aTo
  (\expr',\aenv',\astore', \acont')
  \\
  \pdcfato{(\expr,\aenv)}{\astore}{\epsilon}{((\expr',\aenv'),\astore')}
  & \text{ iff }
  (\expr,\aenv,\astore,\acont) 
  \aTo
  (\expr',\aenv',\astore', \acont')
  \text.
\end{align*}

Now, we can redefine the iteration function (Figure~\ref{fig:widen-trans}).

\begin{figure}
%% \begin{small}
\begin{align*}
  \atf((\hat P,\hat E), \hat H, \astore) &= ((\hat P',\hat E'), \hat H', \astore'')
  \text{, where }
\end{align*}
\begin{align*}
  \hat T_+ &= \setbuild{ (\triedge{\apstate}{\aphrame_+}{\apstate'},\astore') }{
    \pdcfato{\apstate}{\astore}{\aphrame_+}{(\apstate',\astore')}
  }
  \\
  \hat T_\epsilon &= \setbuild{ (\triedge{\apstate}{\epsilon}{\apstate'},\astore') }{
    \pdcfato{\apstate}{\astore}{\epsilon}{(\apstate',\astore')} }
  \\
 \hat T_- &= 
  \big\{
    (\triedge{\apstate''}{\aphrame_-}{\apstate'''},\astore') 
  :
    \pdcfato{\apstate''}{\astore}{\aphrame_-}{(\apstate''',\astore')} \text{ and }  
    \\
    &\hspace{9.5em}
    \triedge{\apstate}{\aphrame_+}{\apstate'} \in \hat E \text{ and }
    \\
    &\hspace{9.5em}
    \biedge{\apstate'}{\apstate''} \in \hat H
    \big\}
  \\
  \hat T' &= {\hat T_+} \union {\hat T_\epsilon} \union {\hat T_-}
  \\
  \hat E' &= \setbuild{ \hat e }{ (\hat e,\_) \in \hat T '}
  \\
  \astore'' &= \bigjoin \setbuild{ \astore' }{ (\_, \astore') \in \hat T'}
  \\
  \hat H_{\epsilon} &= \setbuild{ \biedge{\apstate}{\apstate''} }{
    \biedge{\apstate}{\apstate'} \in \hat H \text{ and } 
      \biedge{\apstate'}{\apstate''} \in \hat H
  }
  \\
  \hat H_{{+}{-}} &= \big\{
  \biedge{\apstate}{\apstate'''}
    : 
    \triedge{\apstate}{\aphrame_+}{\apstate'} \in \hat E
    \text{ and }
    \biedge{\apstate'}{\apstate''} \in \hat H 
    \\
    & \hspace{6.1em} \text{ and } 
    \triedge{\apstate''}{\aphrame_-}{\apstate'''} \in \hat E
    \big\}
  \\
  \hat H' &=   \hat H_{\epsilon} \union   \hat H_{{+}{-}} 
  \\
  \hat P' &= \hat P \union \setbuild{ \apstate' }{ \triedge{\apstate}{\stackact}{\apstate'} } 
  \text.
\end{align*}
\caption{An \ecg{}-powered iteration function for pushdown control-flow analysis with a single-threaded store.}
\label{fig:widen-trans}
%% \end{small}
\end{figure}

\begin{theorem}
Pushdown 0CFA can be computed in $O(n^6)$-time, where $n$ is the
size of the program.
\end{theorem}
\begin{proof}
%\paragraph{Analysis of monovariant complexity}
%
As before, the maximum number of iterations is cubic in the size of
the program for a monovariant analysis.
Fortunately, the cost of each iteration is also now bounded by the number
of edges in the graph, which is also cubic.
\end{proof}

\section{Applications}
\label{sec:applications}

Pushdown control-flow analysis offers more precise control-flow
analysis results than the classical finite-state CFAs.
Consequently, pushdown control-flow analysis improves flow-driven
optimizations (\eg, constant propagation, global register allocation,
inlining~\cite{mattmight:Shivers:1991:CFA}) by eliminating more of
the false positives that block their application.

The more compelling applications of pushdown control-flow analysis
are those which are difficult to drive with classical control-flow
analysis.
Perhaps not surprisingly, the best examples of such analyses are
escape analysis and interprocedural dependence analysis.
Both of these analyses are limited by a static analyzer's ability to
reason about the stack, the core competency of pushdown control-flow
analysis.
(We leave an in-depth formulation and study of these
analyses to future work.)

\subsection{Escape analysis}
In escape analysis, the objective is to determine whether a
heap-allocated object is safely convertible into a stack-allocated
object.
In other words, the compiler is trying to figure out whether the frame
in which an object is allocated outlasts the object itself.
In higher-order languages, closures are candidates for escape
analysis.

Determining whether all closures over a particular \lamterm{} $\lam$
may be heap-allocated is straightforward: find the control states in
the Dyck state graph in which closures over $\lam$ are being created,
then find all control states reachable from these states
over only $\epsilon$-edge and push-edge transitions.
Call this set of control states the ``safe'' set.
Now find all control states which are invoking a closure over $\lam$.
If any of these control states lies outside of the safe set, then
stack-allocation may not be safe; if, however, all invocations lie
within the safe set, then stack-allocation of the closure is safe.

\subsection{Interprocedural dependence analysis}
In interprocedural dependence analysis, the goal is to determine, for
each \lamterm{}, the set of resources which it may read or write when
it is called.
Might and Prabhu showed that if one has knowledge of the program
stack, then one can uncover interprocedural
dependencies~\cite{mattmight:Might:2009:Dependence}.
We can adapt that technique to work with Dyck state graphs.
For each control state, find the set of reachable control states along
only $\epsilon$-edges and pop-edges.
The frames on the pop-edges determine the frames which could have been
on the stack when in the control state.
The frames that are live on the stack determine the procedures that are live on the stack.
Every procedure that is live on the stack has a read-dependence on any
resource being read in the control state, while every procedure that is live
on the stack also has a write-dependence on any resource being written in
the control state.
This logic is the direct complement of ``if $f$ calls $g$ and $g$
accesses $a$, then $f$ also accesses $a$.''

\section{Related work}
\label{sec:related}

Pushdown control-flow analysis draws on work in higher-order
control-flow analysis~\cite{mattmight:Shivers:1991:CFA}, abstract
machines~\cite{mattmight:Felleisen:1987:CESK} and abstract
interpretation~\cite{mattmight:Cousot:1977:AI}.

\paragraph{Context-free analysis of higher-order programs}
The closest related work for this is Vardoulakis and Shivers very
recent work on CFA2~\cite{dvanhorn:DBLP:conf/esop/VardoulakisS10}.
CFA2 is a table-driven summarization algorithm that exploits the
balanced nature of calls and returns to improve return-flow precision
in a control-flow analysis.
Though CFA2 alludes to exploiting context-free languages, context-free
languages are not explicit in its formulation in the same way that
pushdown systems are in pushdown control-flow analysis.
With respect to CFA2, pushdown control-flow analysis is polyvariant,
covers direct-style, and the monovariant instatiation is lower in
complexity (CFA2 is exponential-time).

On the other hand, CFA2 distinguishes stack-allocated and
store-allocated variable bindings, whereas our formulation of pushdown
control-flow analysis does not and allocates all bindings in the
store.
If CFA2 determines a binding can be allocated on the stack, that
binding will enjoy added precision during the analysis and is not
subject to merging like store-allocated bindings.

\paragraph{Calculation approach to abstract interpretation}

\citet{dvanhorn:Midtgaard2009Controlflow} systematically calculate
0CFA using the Cousot-Cousot-style calculational
approach~\citeyearpar{dvanhorn:Cousot98-5} to abstract interpretation
applied to an ANF $\lambda$-calculus.
Like the present work, Midtgaard and Jensen start with the CESK
machine of \citet{mattmight:Flanagan:1993:ANF} and employ a
reachable-states model. 
The analysis is then constructed by composing well-known
Galois connections to reveal a 0CFA incorporating reachability.
The abstract semantics approximate the control stack component of the
machine by its top element.
The authors remark monomorphism materializes in two mappings: ``one
mapping all bindings to the same variable,'' the other ``merging all
calling contexts of the same function.''
Essentially, the pushdown 0CFA of Section~\ref{sec:abstraction}
corresponds to Midtgaard and Jensen's analsysis when the latter
mapping is omitted and the stack component of the machine is not
abstracted.

\paragraph{CFL- and pushdown-reachability techniques}
This work also draws on CFL- and pushdown-reachability
analysis~\cite{mattmight:Bouajjani:1997:PDA-Reachability,mattmight:Kodumal:2004:CFL,mattmight:Reps:1998:CFL,mattmight:Reps:2005:Weighted-PDA}.
For instance, \ecg s, or equivalent variants thereof, appear in many
context-free-language and pushdown reachability algorithms.
For the less efficient versions of our analyses, we implicitly invoked
these methods as subroutines.
When we found these algorithms lacking (as with their enumeration of
control states), we developed Dyck state graph construction.

CFL-reachability techniques have also been used to compute classical
finite-state abstraction CFAs~\cite{mattmight:Melski:2000:CFL} and
type-based polymorphic control-flow
analysis~\cite{mattmight:Rehof:2001:TypeBased}.
These analyses should not be confused with pushdown control-flow
analysis, which is computing a fundamentally more precise kind of CFA.
Moreover, Rehof and Fahndrich's method is cubic in the size of the
\emph{typed} program, but the types may be exponential in the size of
the program.
In addition, our technique is not restricted to typed programs.

\paragraph{Model-checking higher-order recursion schemes}
There is terminology overlap with work by
\citet{mattmight:Kobayashi:2009:HORS} on model-checking higher-order
programs with higher-order recursion schemes, which are a
generalization of context-free grammars in which productions can take
higher-order arguments, so that an order-0 scheme is a context-free
grammar.
Kobyashi exploits a result by \citet{dvanhorn:Ong2006ModelChecking} which
shows that model-checking these recursion schemes is decidable (but
ELEMENTARY-complete) by transforming higher-order programs into
higher-order recursion schemes.
Given the generality of model-checking, Kobayashi's technique may be
considered an alternate paradigm for the analysis of
higher-order programs.
For the case of order-0, both Kobayashi's technique and our own
involve context-free languages, though ours is for control-flow
analysis and his is for model-checking with respect to a temporal
logic.
After these surface similarities, the techniques diverge.
Moreover, there does not seem to be a polynomial-time
variant of Kobyashi's method.

%\todo{Mention Kobayashi work not related.}

\paragraph{Other escape and dependence analyses}
We presented escape and dependence analyses to prove a point: that
pushdown control-flow analysis is more powerful than classical
control-flow analysis, in the sense that it can answer different kinds
of questions.
We have not yet compared our analyses with the myriad escape and
dependence analyses (\eg{}, \cite{mattmight:Blanchet:1998:Escape}) that
exist in the literature, though we do expect that, with their
increased precision, our analyses will be strongly competitive.

% \subsection{Preliminary comparison to classical CFA}

% Figure~\ref{fig:results} describes 
% preliminary
% results of our 
% implementation:
% \begin{center}
%   \verb+http://www.ucombinator.org/projects/pdcfa/+
% \end{center}
% \input{results}

% \section{Future work}

% Linear-bounded automata are more powerful than pushdown
% automata, yet still not Turing-complete.
% %
% It seems likely that a lightweight abstraction could target this
% class to recover even more precision, yet still remain computable.
% %
% It should also be straightforward to handle exceptions in pushdown
% control-flow analysis.

\section{Conclusion}
Pushdown control-flow analysis is an alternative paradigm for
the analysis of higher-order programs.
By modeling the run-time program stack with the stack of a
pushdown system, pushdown control-flow analysis precisely matches
returns to their calls.
We derived pushdown control-flow analysis as an abstract
interpretation of a CESK machine in which its stack component is left
unbounded.
As this abstract interpretation ranged over an infinite
state-space, we sought a decidable method for determining th
reachable states.
We found one by converting the abstracted CESK into a PDA that
recognized the language of legal control-state sequences.
By intersecting this language with a specific regular language and
checking non-emptiness, we were able to answer control-flow questions.
From the PDA formulation, we refined the technique to reduce
complexity from doubly exponential, to best-case exponential, to
worst-case exponential, to polynomial.
We ended with an efficient, polyvariant and precise framework.

\paragraph{Future work}
Pushdown control-flow analysis exploits the fact that clients of
static analyzers often need information about control states rather
than stacks.
Should clients require information about complete
configurations---control states plus stacks---our analysis is lacking.
Our framework represents configurations as \emph{paths} through Dyck
state graphs.
Its results can provide a regular description of the stack, but at a
cost proportional to the size of the graph.
For a client like abstract garbage collection, which would pay
this cost for every edge added to the graph, this cost is
unacceptable.
Our future work will examine how to incrementally summarize stacks
paired with each control state during the analysis.

\paragraph{Acknowledgements}
We thank Dimitrios Vardoulakis for comments and discussions,
and reviewers of ICFP and this workshop.

\bibliographystyle{plainnat}
\bibliography{bibliography}

\end{document}